\documentclass[11pt]{article}       

\usepackage{graphicx}

\usepackage{epsfig}
\usepackage{amsmath}
\usepackage{amsfonts}
\usepackage{amssymb}
\usepackage{enumerate}
\usepackage{graphics}
\usepackage{color}
\usepackage{mathtools}
\usepackage{colordvi}
\usepackage{verbatim}
\usepackage{amsthm}

\newtheorem{theorem}{Theorem}
\newtheorem{example}{Example}
\newtheorem{definition}{Definition}

\newtheorem{lemma}{Lemma}
\newtheorem{remark}{Remark}

\begin{document}

\title{Discovering Small Target Sets in Social Networks:\\ A Fast and Effective Algorithm\thanks{A preliminary version of this paper was presented at the {\em 9th Annual International Conference on Combinatorial Optimization and Applications (COCOA'15)}, December 18-20, 2015,  Houston, Texas, USA.}
}

\author{ G. Cordasco \\ \small Department of  Psychology\\ \small Second University of Naples, Italy
 \and 
 L. Gargano \\ \small Department of Informatics\\ \small University of Salerno, Italy
\and
 M. Mecchia \\ \small Department of Informatics\\ \small University of Salerno, Italy
\and
 A. A. Rescigno \\ \small Department of Informatics\\ \small University of Salerno, Italy
\and   U. Vaccaro\\ \small Department of Informatics\\ \small University of Salerno, Italy}


\date{}

\maketitle

\def\d{{\delta}}
\newcommand{\remove}[1]{}
\newcommand{\N}{{\mathbb{N}}}
\newcommand{\Active}{{{\sf Active}}}
\newcommand{\No}{{\mathbb{N}}_0}


\begin{abstract}
Given a network represented by a graph  $G=(V,E)$, we consider a dynamical process of influence diffusion in $G$ that 
evolves as follows: Initially only the nodes of a given  $S\subseteq V$   are influenced; subsequently, at each round, the set of influenced nodes is augmented by all the nodes in the network  that have a sufficiently large number of already influenced neighbors.
The question is to determine a small subset of nodes $S$ (\emph{a target set}) that  can influence the whole network. 
This is a  widely studied problem that abstracts many phenomena in the social, economic, biological, and physical sciences.  
It is known  
that  the above optimization  problem  is  
hard to approximate within a factor of $2^{\log^{1-\epsilon}|V|}$, for any $\epsilon >0$.
In this paper, we present a  fast  and surprisingly simple algorithm  that exhibits the following features:
1)	when  applied to trees, cycles, or complete graphs,
it  always  produces an optimal solution (i.e, a minimum size target set);
2) when applied to  arbitrary networks, it  always produces a solution of cardinality 
{which improves on 
the previously known upper bound;   
}
3) when applied to real-life networks, it
always produces solutions that substantially outperform the ones obtained by previously published algorithms (for which no proof of optimality or performance guarantee is known in any class of graphs).
\end{abstract}

\newpage

\section{Introduction}

Social networks have been extensively investigated by student of the 
social science  for decades (see, e.g.,  \cite{WF}).
Modern large scale online social networks, like Facebook and LinkedIn,
have made available huge amount of data, thus leading
to many applications of online social networks, and also to  the articulation and 
exploration
of many interesting research questions.
A large part of such studies regards
the analysis of  social influence  diffusion  in  networks of people.
Social influence is the process by which individuals adjust their opinions, 
revise their beliefs, or change their behaviors as a
result of  interactions with other people \cite{CF}.
It has not escaped the attention of advertisers
that the  process of social influence
can be exploited in \emph{viral marketing} \cite{LAM}.
Viral marketing 
refers to the spread of information about products and  behaviors,
and their adoption by people. According to Lately~\cite{D}, ``\emph{the traditional broadcast model of 
advertising-one-way, one-to-many, read-only is increasingly being 
superseded by a vision of marketing that wants, and expects, consumers to spread the word themselves}''.
For what interests   us,  the intent of maximizing the spread of viral information across a network naturally
suggests many interesting  optimization  problems. Some of them  were first articulated in the seminal papers
\cite{KKT-03,KKT-05}.
The recent monograph \cite{BOOK2013} contains an excellent  
description of the area.
In the next section, we will explain and motivate our model of 
information diffusion, state the problem  we are investigating,
describe our results, and discuss how
they relate to the existing literature.
 
\subsection{The Model}

{Let $G = (V,E)$ be a graph modeling the network. We denote by $\Gamma_G(v)$ and  by $d_G(v)=|\Gamma_G(v)|$, respectively, the  neighborhood and the degree  of the vertex $v$ in $G$. Let $t: V \to \No=\{0,1, \ldots  \} $ be a
function assigning  thresholds to the vertices of $G$. For each node $v\in V$, the value 
 $t(v)$  quantifies  how hard it is to influence  node $v$, in the sense that
 easy-to-influence  elements of the network  have ``low'' threshold  values, and
hard-to-influence  elements have  ``high'' threshold  values \cite{Gr}.
}
\begin{definition} 
{Let   $G=(V,E)$  be a graph with threshold function  $t: V \longrightarrow \No$ and  $S\subseteq V$.
An {\em activation process in $G$ starting at $S$}
is a sequence of vertex subsets
$\Active_G[S,0] \subseteq \Active_G[S,1] \subseteq \ldots\subseteq \Active_G[S,\ell] \subseteq \ldots \subseteq V$
of vertex subsets, with} $\Active_G[S,0] = S$ and 
$$\Active_G[S,\ell] = \Active_G[S,\ell-1]\cup \Big\{u :\, \big|\Gamma_G(u)\cap \Active_G[S,\ell - 1]\big|\ge t(u)\Big\},
                                                      \mbox{ for  $\ell \geq 1$.}$$
A \textbf{ target set} for $G$ is  set $S\subseteq V$ such that $\Active_G[S,\lambda]=V$ for some  $\lambda\geq 0$
\end{definition}
In words, at each round  $\ell$ the set of active nodes is
augmented by the set of  nodes $u$ that have a number of
\emph{already} activated neighbors greater or equal to
 $u$'s threshold $t(u)$. 
{The vertex $v$ is said to be {\em  activated} at round $\ell>0$ if $v \in  \Active_G[S,\ell]\setminus \Active_G[S,\ell - 1]$.}

In the rest of the paper we will omit the subscript $G$  whenever the graph $G$  is clear
from the context.

\begin{example}\label{ex-1a}
Consider the tree $T$  in Figure \ref{fig1}. 
The number inside each circle is the  vertex threshold.
A possible   target set for    $T$   is 
$S=\{v_1, v_5, v_7\}$. Indeed we have 

$ 
 \Active[S,0]=S=\{v_1, v_5, v_7\},$

$ 
 \Active[S,1]=S\cup\{v_2,v_3,v_4,v_6,v_8,v_9\},$

 $\Active[S,2]=\Active[S,1]\cup \{   v_{10} \}   =V
$


\begin{figure}
	\centering
	\includegraphics[height=6truecm]{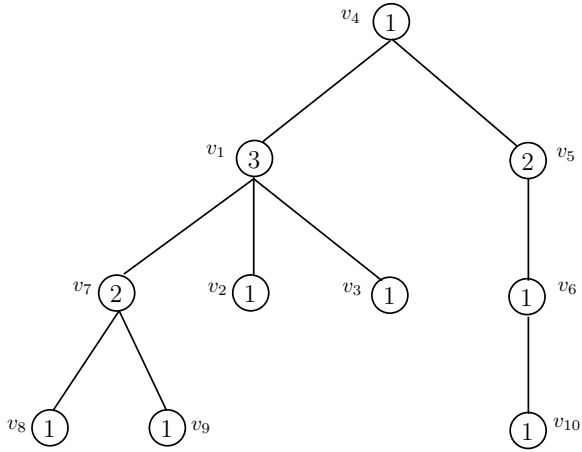}
\caption{A tree with vertex set $V=\{v_1, v_2,\ldots, v_{10}\}$ where 
the number inside each circle is the  vertex threshold.
A    target set       is 
$S=\{v_1, v_5, v_7\}$. \label{fig1}}
\end{figure}

\end{example}

\noindent
The  problem
we  study in this paper
 is defined as follows:
\begin{itemize}
\item[] {\sc Target Set Selection (TSS)}.\\
{\bf Instance:} A network $G=(V,E)$, thresholds $t:V\rightarrow \No$.\\
{\bf Problem:} Find a target set $S\subseteq V$ of \emph{minimum} size for $G$.
\end{itemize}

\subsection{The Context  and our  Results}
The Target Set Selection Problem has roots in the general study
of the \emph{spread of influence} in Social Networks (see  \cite{BOOK2013,EK} and references quoted therein).
For instance, in the area of viral marketing \cite{DR-01}, companies  wanting to
promote products or behaviors might  initially  try to target and convince
a few individuals who, by word-of-mouth, can  trigger
a  cascade of influence in the network leading to
an  adoption  of the products by  a much larger number of individuals.
Recently, viral marketing  has been also recognised as  an important tool in
 the communication strategies of politicians \cite{Bo+,LKLG,T}.

The first authors to study problems of spread of influence in networks
from an algorithmic point of view were Kempe \emph{et al.} \cite{KKT-03,KKT-05}.
However, they were mostly interested in networks with  randomly chosen thresholds.
Chen \cite{Chen-09} studied the following minimization problem:
Given a graph $G$ and fixed arbitrary thresholds $t(v)$, $\forall v\in V$, find
a target set of minimum size that eventually activates
all (or a fixed fraction of) nodes of $G$.
He proved  a strong   inapproximability result that makes unlikely the existence
of an  algorithm with  approximation factor better than  $O(2^{\log^{1-\epsilon }|V|})$.
Chen's result stimulated a series of papers  
\cite{ABW-10,BCNS,BHLM-11,Centeno12,Chopin-12,Chun,Chun2,Cic14,Cic+,C-OFKR,Ga+,NNUW,Re,Za} that isolated interesting cases 
in which the problem (and variants thereof) become tractable.
A notable absence from the literature on the topic (with the exception of \cite{SEP,DZNT})
are heuristics for the Target Set Selection Problem that work for \emph{general graphs}.
This is probably due to the previously quoted strong inapproximability result of Chen \cite{Chen-09},
that seems to suggest that the problem is hopeless.
Providing such an algorithm for general graphs, evaluating its  performances  and esperimentally validating it on real-life networks, is the main objective of this paper.

\subsubsection*{Our Results}

In this paper, we present a  fast  and  simple algorithm  
that exhibits the following features:\\
\textbf{1)}	It always  produces an optimal solution (i.e, a minimum size subset of nodes that influence the whole network) 
in case $G$ is either a tree,  a cycle, or a complete graph.
These results were previously obtained in \cite{Chen-09,NNUW}
by means of \emph{different ad-hoc} algorithms.\\
\textbf{2)}  { For general networks, it  always produces a target set  whose  cardinality
improves on 
the upper bound $ \sum_{v\in V} \min\left(1,\frac{t(v)}{d(v) +1}\right)$  derived  in  \cite{CGM+} and   obtained  in \cite{ABW-10}   by means of the probabilistic method;}\\
\textbf{3)} In real-life networks it produces solutions that outperform the ones obtained using
the algorithms presented  in the papers  \cite{SEP,DZNT}, 
  for which, however,  no proof of optimality or performance guarantee is known in any class of graphs.
	The data sets we use, to experimentally validate our algorithm, include those considered in \cite{SEP,DZNT}.
	
	It is worthwhile  to remark that our algorithm, when executed on a graph $G$ for which the 
	thresholds $t(v)$ have been set equal to the nodes 
	degree $d(v)$, for each $v\in V$, it outputs  a \emph{vertex cover} of $G$, 
	(since in that particular case a target set of $G$ 
	is, indeed, a vertex cover of $G$).
	Therefore, our algorithm appears to be a new algorithm, to the best of our 
	knowledge,
	to compute the vertex cover of graphs (notice that our algorithm differs from
 the classical algorithm that computes a vertex cover  by 
	iteratively deleting a vertex of maximum degree in the graph). We plan to  investigate elsewhere 
	the theoretical performances of our algorithm (i.e., its approximation factor);  computational 
	experiments  suggest that it performs surprisingly well in practice.
	
\section{The TSS algorithm}\label{sec:upper} 
In this section we present our algorithm 
 for the TSS problem.
{The algorithm, given in Figure \ref{fig2},  works by iteratively
 deleting  vertices from the input graph $G$. 
At each iteration, the vertex to be deleted  is  chosen as to
 maximize  a certain  function. During the deletion process,   some vertex $v$ in the surviving graph  may remain with less neighbors than its threshold; in such a case  $v$ is added to the target set     and deleted from the graph while  its neighbors' thresholds are decreased by 1 (since they receive $v$'s influence).
It can also happen that  the surviving graph contains a vertex $v$ whose threshold has been decreased down to 0 (which means that
the deleted nodes are able to activate $v$); 
in such a case $v$ is deleted from the graph
and its neighbors' thresholds are decreased by 1 (since  once $v$ activates,  they will receive $v$'s influence).}

\begin{figure}
\begin{center}
\begin{tabular}{l}
\hline
\hline
{\bf  Algorithm} TSS($G$)\\
{\bf Input:} A graph $G=(V,E)$ with thresholds $t(v)$ for $v\in V$.\\
{\bf \, 1.}  $S=\emptyset$ \\
{\bf \, 2.}   $U=V$ \\  
{\bf \, 3.}   {\bf for } each $v\in V $ {\bf do}  \\
{\bf \, 4.}  $\qquad\quad$  $\d( v)=d( v)$\\
{\bf \, 5.}  $\qquad\quad$  $k( v)=t( v)$ \\
{\bf \, 6.}  $\qquad\quad$  $N( v)=\Gamma( v)$ \\
{\bf \, 7.}  {\bf while} $U\neq \emptyset$  {\bf do}\\
{\bf \, 8.}  $\qquad$ \small{\em [Select one vertex and eliminate it from the graph as specified in the following  cases]}\\
{\bf \, 9.}  $\qquad$ {\bf if } there exists $v\in U$ s.t. $k(v)=0$ {\bf then} \\
{\bf 10.}  $\qquad\quad\quad$  \small{{\em [\underline{Case 1}: The  vertex  $v$ is activated by the influence of its neighbors in $V-U$ only; }}\\
{\bf 11.}  $\qquad\quad\quad$  \small{{\em \hfill it     can then  influence its neighbors in $U$]}}\\
{\bf 12.}  $\qquad\quad\quad$ {\bf for } each $u\in N(v)$ {\bf do} $k(u)=\max(k(u)-1,0)$\\ 
{\bf 13.}  $\qquad$ {\bf else } \\
{\bf 14.}   $\qquad\quad\quad$ {\bf if } there exists $v\in U$ s.t. $\d(v) <  k(v)$ 
                                    {\bf then} \\
{\bf 15.}  $\qquad\quad\quad\quad$ \small{{\em [\underline{Case 2}: The  vertex $v$ is added to $S$, since no sufficient  neighbors remain }}\\
{\bf 16.}  $\qquad\quad\quad\quad \qquad\qquad $ \hfill \small{{\em  in $U$ to activate it;   $v$  can then influence its neighbors in $U$]}}\\
{\bf 17.}  $\qquad\quad\quad\quad $ $S=S\cup\{v\}$ \\ 
{\bf 18.}  $\qquad\quad\quad\quad $ {\bf for } each $u\in N(v)$ {\bf do} $k(u)=k(u)-1$\\ 
{\bf 19.}  $\qquad\quad\quad$ {\bf else }\\
{\bf 20.}  $\qquad\quad\quad\quad $ \small{{\em [\underline{Case 3:} The  vertex $v$ will  be  influenced by some of its  neighbors in $U$]} }\\
{\bf 21.}  $\qquad\quad\quad\quad $ $v={\tt argmax}_{u\in U}\left\{\frac{k(u)}{\delta(u)(\delta(u)+1)}\right\}$ \\
{\bf 22.}  $\qquad$  \small{{\em [Remove the selected vertex $v$ from the graph] } } \\
{\bf 23.}  $\qquad$  {\bf for } each $u\in N(v)$ {\bf do}  \\
{\bf 24.}  $\qquad\qquad\qquad$   $\d(u)=\d(u)-1$\\
{\bf 25.}  $\qquad\qquad\qquad$   $N(u)=N(u)-\{v\}$\\
{\bf 26.}  $\qquad$ $U=U-\{v\}$\\
\hline
\hline
\end{tabular}
\smallskip
\caption{The TSS algorithm. \label{fig2}}
\end{center}
\end{figure}

\noindent
{\bf Example \ref{ex-1a}(cont.)}
{\em Consider the tree $T$  in Figure \ref{fig1}. 
A possible run of  the algorithm $TSS(T)$ removes the nodes from $T$ 
in the  order
of the list below, where  we also 
 indicate for each vertex which among Cases 1, 2, and 3 applies:}
\begin{center}
\begin{tabular}{lllllllllll}
Iteration  \quad & 1\quad & 2\quad & 3\quad & 4\quad & 5\quad & 6\quad & 7\quad & 8\quad & 9\quad & 10
\\
vertex  \quad & $v_{10}$ \quad & $v_9$ \quad & $v_8$ \quad & $v_7$ \quad & $v_6$ \quad & $v_5$ \quad & $v_4$ \quad & $v_3$ \quad & $v_2$ \quad & $v_1$
\\
Case  \quad & 3\quad & 3\quad & 3\quad & 2\quad & 3\quad & 2\quad & 1\quad & 3\quad & 3\quad & 2
 \end{tabular}
\end{center}
{\em Hence the algorithms outputs the set $\{v_1, v_5, v_7\}$ which is a target set for $T$.}

\smallskip
{Before starting the deletion process, the algorithm initialize three variables for each node:
\begin{itemize}
\item $\d(v)$  to the initial degree of node $v$,
\item $k(v)$ to the initial threshold of node $v$, and
\item $N(v)$  the initial set of neighbors of node $v$. 
\end{itemize}
}

In the rest of the paper, we use the following notation.
We denote by $n$ the number of nodes in $G$, that is,  $n=|V|$. Moreover we denote: 
\begin{itemize}
\item By $v_i$ the vertex that is selected during the $n-i+1$-th iteration of the while loop in TSS($G$), for $i=n,\ldots,1$;
\item  by $G(i)$ the graph induced by  $V_i=\{v_{i},\ldots,v_1\}$
\item by $\delta_i(v)$  the value of  $\delta(v)$  as updated at the beginning of the $(n-i+1)-th$ iteration of the while loop in TSS($G$). 
\item by $N_i(v)$  the set  $N(v)$  as updated at the beginning of the $(n-i+1)-th$ iteration of the while loop in TSS($G$), and  
\item by  $k_i(v)$ the value   of $k(v)$ as updated at the beginning of the $(n-i+1)-th$ iteration of the while loop in TSS($G$).
\end{itemize}
For the initial value $i=n$,  the above values are those of the input graph $G$, that is: $G(n)=G$, $\delta_n(v)=d(v)$, $N_n(v)=\Gamma(v)$, $k_n(v)=t(v)$, for each vertex $v$ of $G$.

\medskip
We start with  two technical Lemmata which will be useful in the rest of the paper.
\begin{lemma}\label{fact1}
Consider  a graph $G$. For any $i=n,\ldots,1$ and  
$u\in V_i$, it holds that
\begin{equation}\label{eq-deg}
\Gamma_{G(i)}(u)=N_i(u)  \quad \mbox{ and } \quad d_{G(i)}(u)=\delta_i(u).
\end{equation}
\end{lemma}
\begin{proof}{}
For $i=n$ we have $d_{G(n)}(u)=d_G(u)=\delta_n(u)$ and 
$\Gamma_{G(n)}(u)=\Gamma_{G}(u)=N_n(u)$ for any $u\in V_n=V$. 
\\
Suppose now that the equalities  hold for 
some $i\leq n$. The graph $G(i-1)$ corresponds to the subgraph of $G(i)$ induced
 by $V_{i-1}=V_i-\{v_i\}$. Hence
$$
\Gamma_{G(i-1)}(u)=\Gamma_{G(i)}(u)-\{v_i\},
						$$
								and
	$$d_{G(i-1)}(u)=\begin{cases}{d_{G(i)}(u)-1}&{\mbox {if $u\in\Gamma_{G(i)}(v_i)$,}}\\
                              {d_{G(i)}(u)}&{\mbox {otherwise.}}\end{cases}$$													
We deduce that the desired equalities hold for $i-1$ by noticing that the algorithm uses the same rules to get 											
$$
N_{i-1}(u)=N_i(u)-\{v_i\}$$
and							
$$\delta_{i-1}(u)=\begin{cases}{\delta_{i}(u)-1}&{\mbox {if $u\in N_i(v_i)=\Gamma_{G(i)}(v_i),$}}\\
                              {\delta_{i}(u)}&{\mbox {otherwise}.}\end{cases}$$
\end{proof}

\begin{lemma}\label{lemmateo1}
For any $i>1$, if $S^{(i-1)}$ is a target set for $G(i-1)$ with thresholds   $k_{i-1}(u)$, for $u\in V_{i-1}$,
 then  
\begin{equation}\label{l2}
S^{(i)}=\begin{cases}{S^{(i-1)}\cup \{v_i\}}&{\mbox {if $k_i(v_i)>\delta_i(v_i)$}}\\
                              {S^{(i-1)}}&{\mbox {otherwise}}\end{cases}
\end{equation}
 is a target set for $G(i)$ with thresholds  $k_i(u)$, for $u\in V_{i}$.
\end{lemma}

\begin{proof}
Let us first notice that, according to the algorithm $TSS$, for each $u\in V_{i-1}$ we have
\begin{equation}\label{l1}
k_{i-1}(u)=\begin{cases}{\max(k_i(u)-1,0)}&{\mbox {if $u\in N_i(v_i)$ and  ($k_i(v_i)=0$ or $k_i(v_i)>\delta_i(v_i)$)}}\\
                              {k_i(u)}&{\mbox {otherwise.}}\end{cases}
\end{equation}
\begin{itemize}
\item[1)]
If  $k_i(v_i)=0$, then   
$v_i\in \Active_{G(i)}[S^{(i)},1]$  whatever     $S^{(i)}\subseteq V_i-\{v_i\}$. Hence, by the equation  (\ref{l1}),
  any  target set $S^{(i-1)}$ for $G(i-1)$ is also a target set for $G(i)$.
\item[2)] If $k_i(v_i)>\delta_i(v_i)$ then  $S^{(i)}=S^{(i-1)}\cup\{v_i\}$  and $k_{i-1}(u)=k_i(u)-1$ for each  $u\in N_i(v_i)$. 
It follows that  for any $\ell\geq 0$,
$$\Active_{G(i)}[S^{(i-1)}\cup\{v_i\},\ell]-\{v_i\}=\Active_{G(i-1)}[S^{(i-1)},\ell].$$
Hence,
$\Active_{G(i)}[S^{(i)},\ell]=\Active_{G(i-1)}[S^{(i-1)},\ell]\cup\{v_i\}.$
\item[3)]
Let now $1\leq k_i(v_i)\leq \delta_i(v_i)$.   We have that  $k_{i-1}(u)=k_i(u)$ for each $u\in V_{i-1}$.
If  $S^{(i-1)}$ is a target set for $G(i-1)$,
by definition there exists an integer $\lambda$ such that $\Active_{G(i-1)}[S^{(i-1)},\lambda]=V_{i-1}$. We then have 
$ V_{i-1}\subseteq \Active_{G(i)}[S^{(i-1)},\lambda]$ which implies  $\Active_{G(i)}[S^{(i-1)},\lambda+1]=V_i$.
\end{itemize}
\end{proof}

We can now prove the main result of this section.

\begin{theorem}\label{teo1}
For any graph $G$ and threshold function $t$, the algorithm TSS($G$) outputs a target set for 
$G$.
\end{theorem}
\begin{proof}
Let $S$ be the output of the algorithm $TSS(G)$. We  show that for each $i=1,\ldots,n$ the set 
$S\cap \{v_{i},\ldots,v_1\}$ is a target set for the graph $G(i)$, 
assuming that   each vertex $u$ in $G(i)$ has  threshold $k_i(u)$.
The proof is by induction on the number $i$  of nodes of $G(i)$.
\\
If $i=1$ then  the unique vertex $v_1$ in $G(1)$  either has threshold $k_1(v_1)=0$  and $S\cap \{v_1\}=\emptyset$ or the vertex has positive threshold  $k_1(v_1)>\d_1(v_1)=0$
and   $S\cap \{v_1\}=\{v_1\}$.
\\
Consider now $i>1$ and suppose the algorithm be correct on $G(i-1)$, that is, 
$S\cap \{v_{i-1},\ldots,v_1\}$ is a target set for $G(i-1)$ with threshold function   $k_{i-1}$.
We notice that in each among Cases 1, 2 and 3,   the algorithm updates the thresholds and the target set according to  Lemma \ref{lemmateo1}. Hence, the algorithm is correct on $G(i)$ with threshold function $k_i$.
The theorem follows since $G(n)=G$.
\end{proof}

\noindent
\remove{
\begin{theorem}
For any graph $G$ and threshold function $t$, the algorithm TSS($G$) outputs a solution for the  Target Set Selection  problem  
on $G$.
\end{theorem}
\proof
We  show that for each $i=1,\ldots,n$ the set 
$S\cap \{v_{i},\ldots,v_1\}$ is a target set for the graph $G(i)$, 
assuming that   each vertex $v$ in $G(i)$ has degree $d_{G(i)}(v)$ and threshold $t_{G(i)}(v)$ such that
\begin{equation}\label{assume}
\mbox{  $d_{G(i)}(v)=\delta_i(v), \quad t_{G(i)}(v)=k_i(v), \quad$ and $\quad 0\leq k_i(v)\leq \delta_i(v)$.}
\end{equation}
Notice, that  we admit the possibility of having a node with threshold 0: if $k_i(v)=0$ then $v$ activates independently of $S$ and of its neighbors in $G(i)$,{ that is, $v\in \Active_{G(i)}[S,1]$.} \\
\remove{Summarizing, we are assuming that for each node $v$ in $G(i)$ either 
\begin{itemize}
\item $v$ is in the target set or 
\item $k_i(v)=0$ 
\item or it needs at least $k_i(v)$ active neighbors,  among its $\delta_i(v)$ ones in $G(i)$, to activate. 
In particular,
\begin{itemize}
\item[1)]  $0\leq k_i(v)\leq \delta_i(v)$ means that either $v$ is in the target set or it needs at least $k_i(v)$ active neighbors, 
among its $\delta_i(v)$ ones in $G(i)$, to activate. 
In particular,
\begin{itemize}
\item  $k_i(v)=0$ means that $v$ activates independently of $S$ and of its neighbors in $G(i)$,
{ that is, $v\in \Active_{G(i)}[S,1]$;} 
\item $k_i(v)>\delta_i(v)$ means that $v$ cannot have enough  active neighbor in $G(i)$ and must be 
in each target set for the graph $G(i)$,
{ that is, if $v\not\in S$ then  $v\not\in \Active_{G(i)}[S,\ell]$ for any $\ell\geq 0$};
\end{itemize}
\end{itemize}
}

The proof is by induction on the number $i$  of nodes of $G(i)$.
\\
If $i=1$ then  the unique node $v_1$ in $G(1)$  either has threshold $k_1(v_1)=0$  and $S\cap \{v_1\}=\emptyset$ or the node has positive threshold  $k_1(v_1)>\d_1(v_1)=0$
and   $S\cap \{v_1\}=\{v_1\}$.
\\
Consider now $i>1$ and suppose the algorithm be correct on $G(i-1)$.  
Recall that $v_{i}$ denotes  the node the algorithm eliminates from  $G(i)$ (to obtain $G(i-1)$).
\begin{itemize}
\item
If case 1 of the algorithm holds, then $v_i$ has threshold $k_i(v_i)=0$. 
If case 2 of the algorithm holds, then $v_i$ is added to $S$. 
In both cases,  the node $v_i$ activates independently of the nodes in  $G(i-1)$;  
{ namely, $v\in \Active_{G(i)}[S,1]$;}. 
Moreover, for  each neighbor $u$  of $v_i$ in $G(i)$, the algorithm sets $\delta_{i-1}(u)=\delta_i(u)-1$ and 
 $k_{i-1}(u)=k_i(u)-1$   since the node $u$  gets $v_i$'s influence; thus the assumption that $G(i)$ satisfies  (\ref{assume})
implies that also  $G(i-1)$ satisfies  (\ref{assume}).
We can then use  the inductive hypothesis on $G(i-1)$.  Hence $S\cap \{v_{i-1},\ldots,v_1\}$ is a target set for the graph $G(i-1)$
and we get 
the correctness on $G(i)$ in these  cases.
\item
Consider now case 3 of the algorithm, that is $1\leq k_i(v)\leq \delta_i(v)$ for each node $v$ in $G(i)$.  We have that  $v_i\notin S$.
As above, the graph $G(i-1)$ satisfies (\ref{assume}) and we can apply the inductive hypothesis on  it. Hence, when the initial seed is $S$, all
the neighbors of $v_i$  among $\{v_{i-1},\ldots,v_1\}$ gets active; 
since $\delta_i(v_i)\geq k(v_i)$ also $v_i$ activates in $G(i)$.
\qed
}

It is possible to see that the TSS algorithm can be 
implemented in such a way  to run in $O(|E|\log|V|)$ time. Indeed we need to 
process the nodes $v\in V$ according to the metric $t(v)/(d(v)(d(v)+1))$,
 and the  updates that follow each processed node $v\in V$
involve at most the $d(v)$ neighbors of  $v$.

\section{Estimating the Size of the Solution}
{In this section we prove an upper bound on the size of the target set obtained by the algorithm 
TSS($G$) for any 
input graph $G$. 
Our bound, given in Theorem \ref{teo-upper}, improves on 
the bound  $ \sum_{v\in V} \min\left(1,\frac{t(v)}{d(v) +1}\right)$ given in \cite{ABW-10}  and \cite{CGM+}. Moreover,
the result in \cite{ABW-10} is based on the probabilistic method and an effective
 algorithm   results only by applying   suitable derandomization steps.}

\begin{theorem}\label{teo-upper}
Let $G$ be a connected graph  with at least 3 nodes and  threshold function $t:V\to \No $. 
The algorithm TSS($G$) outputs a target set 
$S$ of size


{{\begin{equation}\label{upper}
|S|\leq 
\sum_{v \in \{u\,|\, u\in V^{(2)}\,  \lor\,  t(u)\neq 1\}}\min\left(1,\frac{t(v)}{d^{(2)}(v){+}1}\right),
\end{equation}
where $V^{(2)}{=}\{v \, |\, v\in V,\, d(v){\geq} 2\}$
and
$d^{(2)}(v){=}|\{u\in \Gamma(v) \, |\, u\in V^{(2)} \, \lor \,  t(u){\neq} 1\}|$.}}

\end{theorem}

\proof
For each $i=1,\ldots, n$, define
\begin{itemize}
\item[a)] {{ $\delta^{(2)}_i(v)=|\{u\in N_i(v) \, |\, u\in V^{(2)} \, \lor \,  t(u)\neq 1\}|$;}}\\
\item[b)]
 $I_i=\left\{v \,{|}\, v\in V_i-V^{(2)}, \,  k_i(v)>\delta_i(v)  \right\}$,  
\item[c)]
$W(G(i))=\sum_{v \in V_i \cap V^{(2)}} \min\left(1,\frac{k_i(v)}{\delta^{(2)}_i(v)+1}.\right)+|I_i|.$
\end{itemize}
We prove that 
\begin{equation}\label{S}
|S\cap V_i|\leq W(G(i)),
\end{equation}
for each $i=1,\ldots, n$. 
The   bound (\ref{upper})  on $S$ follows recalling that   $G(n)=G$ and \\ 
{ {$I_n=\left\{v \,{|}\, v\not\in V^{(2)},\, t(v)= k (v)> \d(v)=d(v)=1 \right\}$}}.

\smallskip

The proof is by induction on $i$.
If $i=1$,  the claim follows  noticing that  
$$ |S\cap \{v_1\}|{=}\begin{cases}{0}&{\mbox{if $k_1(v_1){=}0$}}\\ {1}&{\mbox{if $k_1(v_1){\geq} 1$}}\end{cases}
\quad \mbox{ and } \quad W(G(1)){=}\begin{cases}{0}&{\mbox{if $k_1(v_1){=}0$ and  $v_1{\in} V^{(2)}$}}\\ {1}&{\mbox{otherwise.}}\end{cases}$$

 \noindent
Assume now (\ref{S}) holds  for  $i-1\geq 1$, and  consider  $G(i)$ and 
 the node $v_{i}$. 
We have
$$|S\cap \{v_{i},\ldots,v_1\}|=|S\cap \{v_i\}|+|S\cap \{v_{i-1},\ldots,v_1\}|\leq |S\cap \{v_i\}|+W(G(i-1)).$$
We  show  now that 
$$W(G(i))\geq W(G(i-1)) +|S\cap \{v_i\}|.$$
We first notice that $ W(G(i))-W(G(i-1))$ can be written as 
\begin{eqnarray*}
&& 
\sum_{v\in V_i \cap V^{(2)}} \min\left(1,\frac{k_i(v)}{\delta^{(2)}_i(v)+1}\right){+}|I_i|{-}\sum_{v\in V_{i-1} \cap V^{(2)}} \min\left(1,\frac{k_{i-1}(v)}{\delta^{(2)}_{i-1}(v)+1}\right){-}|I_{i-1}|
\end{eqnarray*}
We notice that  $k_i(v)-1\leq k_{i-1}(v)\leq k_i(v)$ and   $\delta_i(v)-1\leq \delta_{i-1}(v)\leq \delta_i(v)$,
for each neighbor $v$ of $v_i$ in $G(i)$, and that threshold and degree remain unchanged  for each   other
node  in $G(i-1)$. Therefore, we get 
\begin{eqnarray}\label{eq2}
\nonumber W(G(i))-W(G(i-1))\geq  |I_i| &-& |I_{i-1}| 
\\ \nonumber
&& +  \sum_{ v\in N_{i}(v_i) \cap V^{(2)} \atop k_i(v)\leq \delta^{(2)}_i(v)} 
\left(\frac{k_i(v)}{\delta^{(2)}_i(v)+1}-
                  \frac{k_{i-1}(v)}{\delta^{(2)}_{i-1}(v)+1}\right)  
\\ 
&& + \begin{cases} \min\left(1,\frac{k_{i}(v_i)}{\delta^{(2)}_{i}(v_i)+1}\right) 
& {\mbox{if } d(v_i)\geq 2} \\
 0 & \mbox{otherwise}.
\end{cases} 
\end{eqnarray}
We distinguish three cases according to those  in the algorithm TSS($G$).
\begin{description}
\item[I)] Suppose that Case 1 of the Algorithm TSS holds; i.e. $k_i(v_i)=0$. 
Recall that the Algorithm TSS($G$) updates the the values of $\d(u)$ and $k(u)$  for each node in $V_i$ as follows: 
   \begin{equation}\label{update}
	\d_{i-1}(u){=}\begin{cases}{\d_{i}(u){-}1}&{\mbox{if $u {\in} N(v_i)$}}\\
	                         {\d_{i}(u) }&{\mbox{otherwise,}}\end{cases}\quad
	                   k_{i-1}(u){=}\begin{cases}{k_{i}(u){-}1}&{\mbox{if $u {\in} N(v_i)$, $k_i(u){>}0$}}\\
	                         {k_{i}(u) }&{\mbox{otherwise.}}\end{cases}
\end{equation}
By b), (\ref{update}) and being $k_i(v_i)=0$, we immediately get   $I_{i-1}= I_{i}$.
 Hence, from 
(\ref{eq2})  we have
\begin{eqnarray*}
& &W(G(i))-W(G(i-1))\geq\sum_{ v\in N_{i}(v_i) \cap V^{(2)} \atop k_i(v)\leq \delta^{(2)}_i(v)} 
\left(\frac{k_i(v)}{\delta^{(2)}_i(v)+1}-
                  \frac{k_{i-1}(v)}{\delta^{(2)}_{i-1}(v)+1}\right)  \geq 0,
   \end{eqnarray*}
where the last inequality is implied  by (\ref{update}). Since we know that in Case 1 the selected node $v_i$ is not
part of   $S$, we get the desired inequality   $W(G(i))-W(G(i-1))\geq |S\cap \{v_{i}\}|$.

\item[II)] Suppose that Case 2 of the algorithm  holds; i.e. $k_i(v_i)\geq \delta_i(v_i)+1$  and  $k(v)>0$ for each $v\in V_i$. 
The Algorithm TSS($G$) updates the values of $\d(u)$ and $k(u)$  for each node $u\in V_{i-1}$ as in (\ref{update}).
Hence, we have 
$$I_{i-1}=\begin{cases}{I_i}&{\mbox{if $d(v_i)\geq 2$}}\\
                        {I_i-\{v_i\}}&{{\mbox{otherwise}}}\end{cases}$$
and, using this  case assumption, equation  (\ref{eq2}) becomes
\begin{equation*}
W(G(i))-W(G(i-1))\geq 1+
\sum_{v\in N_{i}(v_i) \cap V^{(2)} \atop k_i(v)\leq \delta^{(2)}_i(v)} \left(
\frac{k_i(v)}{\delta^{(2)}_i(v)+1}-\frac{k_{i-1}(v)}{\delta^{(2)}_{i-1}(v)+1}\right) \geq 1.
\end{equation*}
Since in Case 2  $v_i$ is 
part of the output  $S$, we get    $W(G(i))-W(G(i-1))\geq 1=|S\cap \{v_{i}\}|$.


\item[III)] Suppose that Case 3 of the algorithm  holds. 
We know that:
\begin{itemize}
\item[(i)] $1\leq k_i(v)\leq \d_i(v)$, for each $ v \in V_i$; 
\item[(ii)]  $I_i=\emptyset$---by (i) above;	
\item[(iii)] $\frac{k_i(v_i)}{\d_i(v_i)(\delta_i(v_i)+1)}\geq \frac{k_i(v)}{\d_i(v)(\d_i(v)+1)} $,
for each $v\in V_i$;
\item[(iv)]  for each $v\in V_{i-1}$, $k_{i-1}(u)=k_{i}(u)$ and 
	$\d_{i-1}(u)=\begin{cases}{\d_{i}(u){-}1}&{\mbox{ if $u {\in} N(v_i)$}}\\
	                         {\d_{i}(u) }&{\mbox{  otherwise.}}\end{cases}$
\end{itemize}
We distinguish  three cases on the value of  $d(v_i)$ and $\delta_i(v_i)$:
\begin{itemize}
	\item[$\bullet$]
Suppose first  $d(v_i)\geq \delta_i(v_i) \geq 2$.  We have    $\delta_i(v)\geq 2$, for each  $ \ v  \in V_i$.  Otherwise, 
	by (i) we would get $\delta_i(v)=k_i(v)=1$ and, as a consequence 
$$\frac{k_i(v)}{\delta_i(v)(\delta_i(v)+1)}=1/2, \quad \mbox{ while } 
\frac{k_i(v_i)}{\delta_i(v_i)(\delta_i(v_i)+1)} \leq \frac{1}{\delta_i(v_i)+1} \leq 1/3,$$
 contradicting (iii).
Therefore,  by b)   $I_{i-1}=\emptyset$ and 
$\delta^{(2)}_i(v)= \delta_i(v)$, for each $v \in V_i $.
 This, (ii), and (\ref{eq2}) imply
\begin{eqnarray*}
W(G(i))-W(G(i-1))
&\geq&
\sum_{v\in N_{i}(v_i)  \atop k_i(v)\leq \delta_i(v)} \left(
    \frac{k_i(v)}{\delta_i(v)+1} -  \frac{k_{i}(v)}{\delta_{i}(v)}
\right)+ \frac{k_{i}(v_i)}{\delta_{i}(v_i)+1}\\
&=& 
\frac{k_{i}(v_i)}{\delta_{i}(v_i)+1} - 
\sum_{v\in N_{i}(v_i) \atop k_i(v)\leq \delta_i(v)} 
          \frac{k_i(v)}{\delta_i(v)(\delta_i(v)+1)}.
\end{eqnarray*}
As a consequence,  by using (iii) and recalling that $v_i\notin S$ we get
$$
W(G(i))-W(G(i-1))\geq 
          \frac{k_i(v_i)}{\delta_i(v_i)+1}-\frac{k_{i}(v_i)}{\delta_{i}(v_i)+1} = 0  =|S\cap \{v_{i}\}|.$$

\item[$\bullet$] Assume now $d(v_i)\geq 2$ and $\delta_i(v_i) = 1$.
Let $u$ be the   neighbor of $v_i$ in $G(i)$. 
\\
 If $d(u) \geq 2,$  then $u\notin I_{i-1}$ and,
by (ii),  $I_{i-1}=I_i=\emptyset.$ 
By (\ref{eq2}),  we obtain 
\begin{eqnarray*}
W(G(i)){-}W(G(i{-}1))&\geq&   \left(
\frac{k_i(u)}{\delta^{(2)}_i(u)+1}{-}\frac{k_{i-1}(u)}{\delta^{(2)}_{i-1}(u)+1}\right){+}\min\left(1,\frac{k_{i}(v_i)}{\delta^{(2)}_{i}(v_i)+1}\right)\\
&=&   \left(
\frac{k_i(u)}{\delta^{(2)}_i(u)+1}{-}\frac{k_{i}(u)}{\delta^{(2)}_{i}(u)}\right){+}1/2\\
&=&1/2-\frac{k_i(u)}{\delta^{(2)}_i(u)(\delta^{(2)}_i(u)+1)}\\
  &\geq& 1/2-\frac{1}{\delta^{(2)}_i(u)+1}\geq 0 =|S\cap \{v_{i}\}|.
\end{eqnarray*}	
\\
 If {$d(u) =1$  then by (i) $1 \leq k_i(u)\leq t(u) \leq d(u)$ and we have $t(u)=1$.} Moreover, by (iv) $\d_{i-1}(u)=0$,  $\d^{(2)}_i(v_i)=0$ and  $k_{i-1}(u)=k_{i}(u)\geq 1$. Hence $u\in I_{i-1}.$ 
Recalling  that $I_i=\emptyset$, we get $I_{i-1}=\{u\}.$
As a consequence,  (\ref{eq2}) becomes
\begin{eqnarray*}
W(G(i))-W(G(i-1))&\geq&  |I_i|- |I_{i-1}|+ 0 +\min\left(1,\frac{k_{i}(v_i)}{\delta^{(2)}_{i}(v_i)+1}\right) \\
&=& 0 = |S\cap \{v_{i}\}|. 
\end{eqnarray*}

	\item[$\bullet$] Suppose finally   $d(v_i)=1$. Let $u$ be the unique neighbor of $v_i$ in $G(i)$
 \\
If $d(u)\geq 2$, then    $u\notin I_{i-1}$ and,
by (ii),  $I_{i-1}=I_i=\emptyset.$  Moreover, {by (i) we know that $1 \leq k_i(v_i)\leq t(v_i) \leq d(v_i)$ and we have 
$t(v_i)=1$. Hence } $\delta^{(2)}_i(u)=\delta^{(2)}_{i-1}(u)$.  By  (\ref{eq2}),  we obtain
\begin{equation*}
W(G(i)){-}W(G(i{-}1))\geq  0+ \left(\frac{k_i(u)}{\delta^{(2)}_i(u)+1}{-}\frac{k_{i-1}(u)}{\delta^{(2)}_{i-1}(u)+1}\right)
  = 0 =|S\cap \{v_{i}\}|.
\end{equation*}	
 Finally, the case  $d(u)\leq 1$ can hold  only if  the 
input graph  $G$ has a connected component  consisting of  two nodes. This is excluded by the theorem hypothesis.  
\end{itemize}
\qed
\end{description}

\medskip
\begin{remark}
We notice that the  bound in Theorem \ref{teo-upper} improves on  the previously known bound  $\sum_{v\in V} \min\left(1, t(v)/(d(v)+1)\right)$ given in 
\cite{ABW-10,CGM+}.
Indeed we are able to show that  for any graph 

\begin{equation} \label{boundComparison}
\sum_{v \in \{u\,|\, u\in V^{(2)}\,  \lor\,  t(u)\neq 1\}}\min\left(1,\frac{t(v)}{d^{(2)}(v){+}1}\right)
\leq 
\sum_{v\in V} \min\left(1, \frac{t(v)}{d(v)+1} \right).
\end{equation}

In order to prove  (\ref{boundComparison}),  we first notice that  the difference between the two bounds can be written as, 

\begin{eqnarray*}
 && \sum_{v\in V} \min\left(1, \frac{t(v)}{d(v)+1} \right) -  \sum_{v \in \{u\,|\, u\in V^{(2)}\,  \lor\,  t(u)\neq 1\}}\min\left(1,\frac{t(v)}{d^{(2)}(v){+}1}\right)=\\
&&\sum_{v\in V^{(2)}} \min\left(1, \frac{t(v)}{d(v){+}1} \right){+}\sum_{v\notin V^{(2)}} \min\left(1, \frac{t(v)}{2} \right) {-}   \sum_{v\in V^{(2)}} \min\left(1,\frac{t(v)}{d^{(2)}(v){+}1}\right) {+} \sum _{v\notin V^{(2)}\atop t(v)>1 } {1}=\\
&& \sum_{v\in V^{(2)}} \min\left(1, \frac{t(v)}{d(v){+}1} \right) {+} \sum_{v \notin V^{(2)}\atop  t(v)=1} {1/2}  -  \sum_{v\in V^{(2)}} \min\left(1,\frac{t(v)}{d^{(2)}(v){+}1}\right) \geq\\
&&\sum_{v\in V^{(2)} \atop t(v)\leq d(v)} \frac{t(v)}{d(v)+1} + \sum_{v \notin V^{(2)} \atop  t(v)=1} {1/2}  -   \sum_{v\in V^{(2)} \atop t(v)\leq d(v)} \min\left(1,\frac{t(v)}{d^{(2)}(v) +1}\right) \geq \\
&&\sum_{v\in V^{(2)} \atop t(v)\leq d(v)} \left(\frac{t(v)}{d(v)+1} +  {\frac{d(v)-d^{(2)}(v)}{2}} \right)  -    \sum_{v\in V^{(2)} \atop t(v)\leq d(v)} \min\left(1,\frac{t(v)}{d^{(2)}(v) +1}\right), 
\end{eqnarray*}

where the last inequality is due to the fact that 
$$\sum_{v \notin V^{(2)}\atop  t(v)=1} {1/2} = \sum_{v\in V^{(2)} } {\frac{d(v)-d^{(2)}(v)}{2}} \geq \sum_{v\in V^{(2)} \atop t(v)\leq d(v)} {\frac{d(v)-d^{(2)}(v)}{2}} $$
that is,  we are aggregating  the contribution of each node, having both degree and threshold equal to $1$, to that of  its 
unique neighbor.

Now let us consider the contribution of each $v \in V^{(2)}$, such that $t(v)\leq d(v)$, to the equation above.  If $d(v)=d^{(2)}(v)$, then clearly the contribution of $v$ is zero. If $d(v)-d^{(2)}(v)\geq 2$ then the contribution of $v$ is 
$$\frac{t(v)}{d(v)+1} +  {\frac{d(v)-d^{(2)}(v)}{2}} -     \min\left(1,\frac{t(v)}{d^{(2)}(v) +1}\right) \geq \frac{t(v)}{d(v)+1}+1-1\geq0$$

Finally, if $d(v)-d^{(2)}(v)=1$ we have 
$$\frac{t(v)}{d(v)+1} +  {1/2} -     \min\left(1,\frac{t(v)}{d(v)}\right)=\frac{t(v)}{d(v)+1} +  {1/2} -     \frac{t(v)}{d(v)} = \frac{2(d(v)-t(v))}{2d(v)(d(v)+1)} \geq0.$$
In each   case  the contribution of $v$ is non negative and   (\ref{boundComparison}) holds.
\medskip

Furthermore it is worth to notice that  our bound can give a dramatic improvement 
with respect to one in \cite{ABW-10,CGM+}.
As an example consider the star graph  on $n$ nodes 
with center $c$ given in Figure \ref{figureStar}
and thresholds equal to  1 for each leaf node and to   $t(c)\leq n$ for the center node $c$.
The ratio of the bound in \cite{ABW-10,CGM+} to the one in this paper is 
$$\frac{\sum_{v\in V} \min\left(1, \frac{t(v)}{(d(v)+1)} \right)}
   {\sum_{v \in \{u\,|\, u\in V^{(2)}\,  \lor\,  t(u)\neq 1\}}\min\left(1,\frac{t(v)}{d^{(2)}(v){+}1}\right)} 
=
\frac{\frac{t(c)}{n}+\frac{n-1}{2}}{1+0}\geq  \frac{n-1}{2}.$$
\end{remark}

\begin{figure}[h!]
	\centering
		\includegraphics[height=2.5truecm]{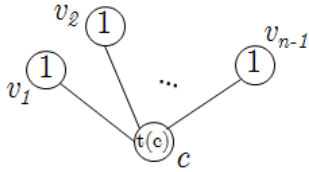}
		\caption{{A star graph with $n$ nodes. The bound in \cite{ABW-10,CGM+} provides a target set of size   $\frac{t(c)}{n}+\frac{n-1}{2}$ while the bound in Theorem \ref{teo-upper} is $1$. In this specific case the bound of Theorem \ref{teo-upper} is tight, the optimal target set consists of the center node $c$.} 
 \label{figureStar}}
\end{figure}

\remove{ 
\proof
Let $W(G(i))=\sum_{j=1}^i \min\left(1,\frac{k_i(v_j)}{\delta_i(v_j)+1}\right)$. We prove 
by induction on $i$ that  
\begin{equation}\label{S}
|S\cap \{v_{i},\ldots,v_1\}|\leq W(G(i)).
\end{equation}
The   bound (\ref{upper})  on $S$ follows recalling that   $G(n)=G$.
If $i=1$  we have
$|S\cap \{v_1\}|=
\min\left(1,\frac{k_1(v_1)}{\delta_1(v_1)+1}\right)=W(G(1))$.
 Assume now (\ref{S}) holds  for  $i-1\geq 1$, and  consider  $G(i)$ and 
 the node $v_{i}$. 
We have 
$|S\cap \{v_{i},\ldots,v_1\}|=|S\cap \{v_i\}|+|S\cap \{v_{i-1},\ldots,v_1\}|\leq |S\cap \{v_i\}|+W(G(i-1)).$
\noindent
We  show  that $W(G(i))\geq W(G(i-1)) +|S\cap \{v_i\}|$.
Recalling that  $N_i(v_i)$ denotes the neighborhood of $v_i$ in $G(i)$, we have
\begin{eqnarray*}
&&W(G(i))-W(G(i-1))=  \nonumber\\
&& \quad =
\sum_{j=1}^i \min\left(1,\frac{k_i(v_j)}{\delta_i(v_j)+1}\right)-\sum_{j=1}^{i-1} \min\left(1,\frac{k_{i-1}(v_j)}{\delta_{i-1}(v_j)+1}\right)
\nonumber\\
&& \quad =
\min\left(1,\frac{k_{i}(v_i)}{\delta_{i}(v_i)+1}\right)+
\sum_{v\in N_{i}(v_i)} \left[
         \min\left(1,\frac{k_i(v)}{\delta_i(v)+1}\right) - \min\left(1,\frac{k_{i-1}(v)}{\delta_{i-1}(v)+1}\right)
                  \right]\nonumber
\end{eqnarray*}
Therefore, we get
\begin{equation}\label{eq2}
W(G(i))-W(G(i-1))=
\min\left(1,\frac{k_{i}(v_i)}{\delta_{i}(v_i)+1}\right) 
+
\sum_{ v\in N_{i}(v_i) \atop k_i(v)\leq \delta_i(v)} \left[
\frac{k_i(v)}{\delta_i(v)+1}-\frac{k_{i-1}(v)}{\delta_{i-1}(v)+1}\right]
\end{equation}
We distinguish three cases according to the cases  in the algorithm TSS($G$).
\begin{itemize}
\item Suppose that Case 1 of the Algorithm TSS holds; i.e. $k_i(v_i)=0$. 
By (\ref{eq2}), we get 
\begin{equation*}
W(G(i))-W(G(i-1))=
\sum_{v\in N_{i}(v_i) \atop k_i(v)\leq \delta_i(v)} \left[
\frac{k_i(v)}{\delta_i(v)+1}-\frac{k_{i}(v)-1}{\delta_{i}(v)}\right]\geq 0= |S\cap \{v_{i}\}|.
   \end{equation*}

\item Suppose that Case 2 of the algorithm  holds; i.e. $k_i(v_i)\geq \delta_i(v_i)+1$. 
 By (\ref{eq2}), we get 
\begin{equation*}
W(G(i))-W(G(i-1))=
1 + \sum_{v\in N_{i}(v_i) \atop k_i(v)\leq \delta_i(v)} \left[
\frac{k_i(v)}{\delta_i(v)+1}-\frac{k_{i}(v)-1}{\delta_{i}(v)}\right] \geq 1 =|S\cap \{v_{i}\}|.
\end{equation*}

\item Suppose that Case 3  holds; i.e. $k_i(v_i)\leq \delta_i(v_i)$. 
In such a case we know that 
$$\frac{k_i(v)}{\delta_i(v)(\delta_i(v)+1)} \leq \frac{k_i(v_i)}{\delta_i(v_i)(\delta_i(v_i)+1)}$$
for each $v\in \{v_i,\ldots,v_1\}$
and
$S\cap \{v_{i}\}=\emptyset$. By  this and (\ref{eq2}),  we obtain 
\begin{eqnarray*}
W(G(i))-W(G(i-1))
&=&
\frac{k_{i}(v_i)}{\delta_{i}(v_i)+1}  + 
\sum_{v\in N_{i}(v_i) \atop k_i(v)\leq \delta_i(v)} \left[
    \frac{k_i(v)}{\delta_i(v)+1} -  \frac{k_{i}(v)}{\delta_{i}(v)}
\right]\\
&=& 
\frac{k_{i}(v_i)}{\delta_{i}(v_i)+1} - 
\sum_{v\in N_{i}(v_i) \atop k_i(v)\leq \delta_i(v)} 
          \frac{k_i(v)}{\delta_i(v)(\delta_i(v)+1)} \\
&\geq &
          \frac{k_i(v_i)}{\delta_i(v_i)+1}-\frac{k_{i}(v_i)}{\delta_{i}(v_i)+1} \\
&=& 0  \\
&=&|S\cap \{v_{i}\}|. 
\end{eqnarray*}
\end{itemize}
\qed
}
\section{Optimality Cases}\label{sec:trees} 
In this section, we prove that our algorithm TSS provides a unified setting for  several results,
  obtained in the literature  by means of different \emph{ad hoc} algorithms.
Trees, cycles and cliques are among the few cases known to admit optimal polynomial 
time algorithms  for the TSS problem \cite{Chen-09,NNUW}. 
In the following, we prove that our algorithm TSS provides the \emph{first}  unifying setting for all  these cases.
\begin{theorem} \label{trees}
The algorithm  TSS($T$) returns an optimal solution   whenever the input graph  $T$ is a  tree.
\end{theorem}
\proof
Let $T=(V,E)$ and  $n=|V|$.
We recall that for $i=1,\ldots,n$: $v_i$ denotes the node selected during the $n-i+1$-th iteration of the while loop in TSS, 
 $T(i)$ is the forest induced by the   set  $V_i=\{v_{i},\ldots,v_1\}$, and
 $\delta_i(v)$ and  $k_i(v)$ are the degree and threshold of $v$, for $v\in V_i$.
Let $S$ be the target set produced by the algorithm TSS($T$). We prove by induction on $i$ that 
\begin{equation}\label{eq-T}
 |S\cap\{v_{i},\ldots,v_1\}|=|S^*_i|,
\end{equation}
where $S^*_i$ represents an optimal target set for the forest $T(i)$  with  threshold function  $k_i$.
For $i=1$, it is immediate  that  for the only node $v_1$ in $F(1)$ one has 
$$S\cap\{v_1\} = S^*_1 = 
    \begin{cases}{\emptyset} & {\mbox{ if } k_1(v_1)=0}\cr{\{v_1\}}& {\mbox{ otherwise. }}\end{cases}$$
Suppose now (\ref{eq-T}) true for $i-1$ and consider  the tree $T(i)$ and the selected node $v_i$.
\begin{enumerate}
\item Assume first that $k_i(v_i)=0$. We get 
$$|S\cap\{v_{i},\ldots,v_1\}|=|S\cap\{v_{i-1},\ldots,v_1\}|=|S^*_{i-1}|\leq |S^*_{i}|$$ and
the equality (\ref{eq-T}) holds for $i$.
\item Assume now that $k_i(v_i)\geq \delta_i(v_i)+1$. 
Clearly, any solution for $T(i)$ 
must include node $v_i$, otherwise it cannot be activated.
This implies that 
$$|S^*_i|=1+|S^*_{i-1}|=1+ |S\cap\{v_{i-1},\ldots,v_1\}|=|S\cap\{v_{i},\ldots,v_1\}|$$
and (\ref{eq-T}) holds for $i$.
\item
Finally, suppose  that $v_i={\tt argmax}_{i\geq j\geq 1}\left\{{k_i(v_j)}/{(\delta_i(v_j)(\delta_i(v_j)+1))}\right\}$ (cfr. line 21 of the algorithm).
In this case each leaf $v_j$ in $T(i)$ has 
$$\frac{k_i(v_\ell)}{\delta_i(v_\ell)(\delta_i(v_\ell)+1)}= \frac{1}{2}$$
 while
each internal node 
$v_\ell$ has 
$$\frac{k_i(v_\ell)}{\delta_i(v_\ell)(\delta_i(v_\ell)+1)}\leq \frac{1}{\delta_i(v_\ell)+1}\leq \frac{1}{3}.$$
Hence,  the  node $v_i$
must be  a leaf in $T(i)$ and has $k_i(v_i)=\delta_i(v_i)=1$.
Hence $|S\cap\{v_{i},\ldots,v_1\}|=|S\cap\{v_{i-1},\ldots,v_1\}|=|S^*_{i-1}|\leq |S^*_{i}|$.
\hphantom{aaaaaaaaa}\qed
\end{enumerate}

\begin{theorem}\label{teoC}
 The algorithm  TSS($C$) outputs an optimal solution   whenever the input graph   $C$ is  a cycle.
\end{theorem}
\proof 
If  the first selected node $v_n$ has threshold    0  then clearly $v_n\not \in S^*$ for any optimal 
solution $S^*$.\\
If the threshold of  $v_n$ is larger than
 its degree then clearly $v_n\in S^*$ for any optimal 
solution $S^*$.
In both cases  $v_n\in \Active[S^*,1]$ and its neighbors can use $v_n$'s influence; 
that is, the algorithm correctly sets $k_{n-1}=\max(k_n -1,0)$ for these two nodes.\\
If threshold of each node $v\in V$ is $1\leq t(v)\leq d(v)$, we get that  during the first iteration of the algorithm  TSS($C$), the selected  node $v_n$ satisfies Case 3  and has  $t(v_n)=2$ if at least one of the nodes in $C$ has threshold $2$, otherwise  $t(v_n)=1$.
Moreover, it is not difficult  to see that there exists an optimal solution $S^*$ for $C$ such that $S^*\cap\{ v_n\}=\emptyset$. 
\remove{This is obvious if  either $t(v)=1$ for each $v$  or $t(v)=2$ for each $v$.
In case both 1 and 2 appear as threshold of some node in $C$, then  one optimal solution can be  formed  by  

-- $\lfloor\frac{|\{v\ | t(v)=2\}|}{2}\right\rfloor$
nodes of threshold 2  (just start from any node, move clockwise along the cycle and take each second one node having threshold 2) and, 

-- if  $|\{v\ | t(v)=2\}|$ is odd, 1 node of threshold 1.
}
\\
In each case, the result   follows by Theorem \ref {trees}, since the remaining graph is a path on nodes $v_{n-1},\ldots ,  v_1$.
\qed
\remove{
If  the first selected node $v_n$ has threshold    0  then clearly $v_n\not \in S^*$ for any optimal 
solution $S^*$.\\
If the threshold of  $v_n$ is larger than
 its degree then clearly $v_n\in S^*$ for any optimal 
solution $S^*$.
In both cases  $v_n\in \Active[S^*,1]$ and its neighbors can use $v_n$'s influence; 
that is, the algorithm correctly sets $k_{n-1}=\max(k_n -1,0)$ for these two nodes.\\
If threshold of each node $v\in V$ is $1\leq t(v)\leq d(v)$, we get that  during the first iteration of the algorithm  TSS($C$), the selected  node $v_n$ satisfies Case 3  and has  $t(v_n)=2$ if at least one of the nodes in $C$ has threshold $2$, otherwise  $t(v_n)=1$.
Moreover, it is not difficult  to see that there exists an optimal solution $S^*$ for $C$ such that $S^*\cap\{ v_n\}=\emptyset$. 
\remove{This is obvious if  either $t(v)=1$ for each $v$  or $t(v)=2$ for each $v$.
In case both 1 and 2 appear as threshold of some node in $C$, then  one optimal solution can be  formed  by  

-- $\lfloor\frac{|\{v\ | t(v)=2\}|}{2}\right\rfloor$
nodes of threshold 2  (just start from any node, move clockwise along the cycle and take each second one node having threshold 2) and, 

-- if  $|\{v\ | t(v)=2\}|$ is odd, 1 node of threshold 1.
}
\\
In each case, the result   follows by Theorem \ref {trees}, since the remaining graph is a path on $v_{n-1},\ldots v_1$.
\qed
END REMOVE }

\begin{theorem}\label{teoK} Let $K=(V,E)$ be a   clique with  $V= \{u_1,\ldots, u_n\}$ and $t(u_1)\leq\ldots\leq t(u_{n-m})< n \leq t(u_{n-m+1})\leq \ldots \leq t(u_n)$. The algorithm  TSS($K$) outputs an optimal target set of size 
\begin{equation}\label{eq-SK} 
m + \quad \max_{\mathclap{1\leq j\leq n-m}} \quad   \max(t(u_j) -m -j+1, 0).
\end{equation}
\end{theorem}

\proof
It is well known  that there exists an optimal target set $S^*$ consisting of the $|S^*|$ nodes of higher threshold \cite{NNUW}.
Being $S^*$  a target set,  we know that  each node in the graph $K$ must activate. Therefore, for each $u\in V$  there exists some iteration  $i\geq 0$ 
such that $u\in\Active[S,i]$.
Assume $V= \{u_1,\ldots, u_n\}$ and 
     $$t(u_1)\leq\ldots\leq t(u_{n-m})< n \leq t(u_{n-m+1})\leq \ldots \leq t(u_n).$$
  Since the thresholds are non decreasing with the node index, it follows that:  
\begin{itemize}
\item  for each $j\geq n-m+1$, the node $u_j$ has threshold
 $ t(u_j)\geq n$ and  $u_j\in S^*$ must hold. Hence,  $|S^*|\geq m$;
\item for each  $j\leq n-|S^*|$,
the node $u_j$ activates if it gets, in addition to  the influence of its
$m$ neighbors with threshold larger than $n-1$, the influence of  at least $t(u_j)-m$ other neighbors, hence we have  
that  $$t(u_j)-m\leq j-1+(|S^*|-m)$$
 must hold;
\item for each  $j=n-|S^*|+1, \ldots, n-m$, we have 
$$t(u_j)-m-j+1\leq (n-1) - m - (n-|S^*|+1) +1=|S^*|-m +1.$$
\end{itemize}
Summarizing, we get, 

$$|S^*|\geq m + \quad  \max_{{1\leq j\leq n-m}}\quad  \max\left(t(u_j) -m -j +1,0\right).$$
We show now that the algorithm TSS outputs a target set $S$ whose size 
is upper bounded by the value in  (\ref{eq-SK}). Notice that, in general, the output  $S$ does not consist of the 
nodes having the highest thresholds.

Consider the residual graph $K(i)=(V_i,E_i)$, for some $1\leq i\leq n$. 
It is easy to see that   for any $u_j, u_s\in V_i$ it holds
  \begin{itemize}
	\item[] 1) $\delta_i(u_j)=i$;
	\\
	2) if $j<s$ then  $k_i(u_j)\leq k_i(u_s)$;
	\\
	3) if $t(u_j) \geq n$ then  $k_i(u_j) \geq i$,
\\
4) if $t(u_j) < n$ then  $k_i(u_j) \leq i$.
\end{itemize}
W.l.o.g. we assume that at any iteration of algorithm TSS if the node to be selected is not unique then  the tie is broken as follows (cfr. point 2) above):  
\begin{itemize}
	\item[i)] If Case 1 holds then the selected node is the one  with the lowest index,
 \item[ii)] if either Case 2 or Case 3 occurs then 
      the selected node is the one  with the largest index.
\end{itemize}
\noindent
Clearly, this implies that $K(i)$ contains $i$ nodes with consecutive indices among $u_1,\ldots,u_n$,
that is, 
\begin{equation} \label{Vi}
V_i=\left\{u_{\ell_i},u_{\ell_i+1},\ldots,u_{r_i}\right\}
\end{equation}
for some $\ell_i\geq 1$ and $r_i=\ell_i+i-1$.

Let $h=n-m$. 
We shall  prove by induction on $i$ that, for each $i=n,\ldots, 1$, 
at the beginning of the $n-i+1$-th iteration of the while loop in TSS($K$), it holds 
\begin{equation} \label{UBK2}
|S\cap V_i| \leq 
\begin{cases}
{(r_i-h)+\max_{\substack{\ell_i \leq j \leq h}} \max( k_i(u_j)-(r_i-h)-j+\ell_i, \, 0)} & {\mbox{ if $r_i > h$,}}\\
{\max_{\substack{\ell_i \leq j \leq r_i} } \max(k_i(u_j)-j+\ell_i ,\, 0)} &{\mbox{ if $r_i \leq h$.}}
\end{cases} 
\end{equation}
The  upper bound (\ref{eq-SK}) follows when  $i=n$; indeed $K(n)=K$ and $|S|=|S\cap V(n)|$.
\\
For $i=1$, $K(1)$ is induced by only one node, let say $u$, and
$$|S\cap\{u\}| =
\begin{cases}
{1} & {\mbox{\ \ if $k_1(u)\geq 1$,}} \\
{0}  & {\mbox{\ \ if $k_1(u)=0$.}} 
\end{cases}
$$
proving that the bound holds in this case.\\
Suppose now (\ref{UBK2}) true for some $i-1\geq 1$ and consider the $n-i+1$-th iteration of the algorithm TSS. 
Let $v$ be the node selected by algorithm TSS at the $n-i+1$-th iteration.
We distinguish three cases according to the  cases of the algorithm TSS($G$).

\medskip
\noindent
{\underline{Case 1:
 $k_{i}(v)=0$}. \  By i) and (\ref{Vi}),  one has  $v=u_{\ell_i}$, $\ell_{i-1}=\ell_i+1$ and $r_{i-1}=r_i$.
Moreover, $k_{i}(u_j)=k_{i-1}(u_j)+1$  for each $u_j\in V_{i-1}$.
Hence, 
\begin{eqnarray*}
\lefteqn{|S\cap V_i| = |S\cap V_{i-1}|}\\
 & & \leq 
\begin{cases}
{(r_i{-}h)+\max_{\substack{\ell_i+1 \leq j \leq h}} \max( k_{i-1}(u_j){-}(r_i{-}h){-}j+\ell_i{+}1,\, 0)} & {\mbox{if $r_i > h$,}}\\ 
{\max_{\substack{\ell+1 \leq j \leq r} } \max(k_{i-1}(u_j)-j+\ell+1,\, 0)} & {\mbox{if $r_i \leq h$,}}
\end{cases} 
\\
&& \\
& & = 
\begin{cases}
{(r_i-h)+\max_{\substack{\ell_i \leq j \leq h}} \max(k_{i}(u_j)-(r_i-h)-j+\ell_i,0)}& {\mbox{if $r_i > h$,}}\\
{\max_{\substack{\ell \leq j \leq r} } \max(k_{i}(u_j)-j+\ell , 0)} & {\mbox{ if $r \leq h$.}}
\end{cases}
\end{eqnarray*}

\medskip
\noindent
{\underline{Case 2: $k_{i}(v)>\delta_{i}(v)$}. \ By  ii) and (\ref{Vi}) we have $v=u_{r_i}$, $\ell_i=\ell_{i-1}$, $r_{i-1}=r_i-1$.
Moreover,  $k_{i}(u_j)=k_{i-1}(u_j)+1$  for each $u_j\in V_{i-1}$.
Recalling relations 3) and 4),
we have 
\begin{eqnarray*}
\lefteqn{|S\cap V_i| =  1+|S\cap V_{i-1}|}\\
 & & \leq 1{+} 
\begin{cases}
(r_{i-1}{-}h){+}\max_{\substack{\ell_{i-1} \leq j \leq h}} \max(k_{i-1}(u_j){-}(r_{i-1}{-}h){-}j{+}\ell_{i-1},\, 0) & \mbox{if $r_{i-1}{>} h$,}\\  
\max_{\substack{\ell_{i-1} \leq j \leq r_{i-1}} } \max(k_{i-1}(u_j)-j+\ell_{i-1},\, 0) & \mbox{if $r_{i-1} {\leq} h$,}
\end{cases} \\
&& \\
& & = 
\begin{cases}
{(r_i{-}h)+\max_{\substack{\ell_i \leq j \leq h}} \max(k_{i-1}(u_j)+1{-}(r_i-h){-}j+\ell_i,\, 0)} & {\mbox{if $r_i{-}1 > h$,}}\\
{\max_{\substack{\ell \leq j \leq r_i-1} } \max(k_{i-1}(u_j)+1-j+\ell_i ,\, 1) } & { \mbox{if $r_i{-}1 \leq h$.}}
\end{cases} \\
& &  \\
& & = 
\begin{cases}
{(r_i-h)+\max\{0,\max_{\substack{\ell_i \leq j \leq h}} k_{i}(u_j)-(r_i-h)-j+\ell_i\}} & { \mbox{\ \ if $r_i > h$,}}\\
{\max \{0,\max_{\substack{\ell_i \leq j \leq r_i} } k_{i}(u_j)-j+\ell_i \}} & { \mbox{\ \ if $r_i \leq h$.}}
\end{cases} 
 \end{eqnarray*}

\medskip
\noindent
{\underline{Case 3: $0<k_{i}(v)\leq \delta_{i}(v)$}. 
\   By  ii) and (\ref{Vi}) we have $v=u_{r_i}$, $\ell_i=\ell_{i-1}$, $r_{i-1}=r_i-1$.
Moreover,  $k_{i}(u_j)=k_{i-1}(u_j)$  for each $u_j\in V_{i-1}$.
Recalling that by 3) and 4)  we have $t(u_r) <n$, which implies  $r_i\leq h$,
we have 
\begin{eqnarray*}
|S\cap V_i| = |S\cap V_{i-1}|
 & & \leq \max_{\substack{\ell_{i-1} \leq j \leq r_{i-1}} } \max(k_{i-1}(u_j)-j+\ell_{i-1}, \, 0)  \\
 & & \leq \max_{\substack{\ell_i \leq j \leq r_i-1} } \max(k_{i}(u_j)-j+\ell_i, \, 0) \\
 & & \leq \max_{\substack{\ell_i \leq j \leq r_i} } \max(k_{i}(u_j)-j+\ell_i, \, 0). 
\end{eqnarray*}
\qed

\remove{
\proof {\em (Sketch - the full proof is given in  Appendix A.)}
 It is well known  that there exists an optimal target set $S^*$ consisting of the $|S^*|$ nodes of higher threshold \cite{NNUW}.
Being $S^*$  a target set,   each node $u_j$ must activate, that is, $u_j\in\Active[S,i]$ for some $i\geq 0$.
Assume $V= \{u_1,\ldots, u_n\}$ and $t(u_1)\leq\ldots\leq t(u_{n-m})< n \leq t(u_{n-m+1})\leq \ldots \leq t(u_n)$.  Since the thresholds are non decreasing with the node index, it follows that:  
\begin{itemize}
\item  for each of the $m$ nodes s.t. $ t(u_j)\geq n$, it must hold $u_j\in S^*$, hence  $|S^*|\geq m$;
\item for each  $j\leq n-|S^*|$,
the node $u_j$ activates if it gets, in addition to  the influence of its
$m$ neighbors with threshold larger than $n-1$, the influence of  $t(u_j)-m$ other neighbors, hence we have  
that  $t(u_j)-m\leq j-1+(|S^*|-m)$ must hold;
\item if $n-|S^*|+1  \leq j\leq n-m$, then
$t(u_j)-m-j+1\leq (n-1) - m - (n-|S^*|+1) +1.$
\end{itemize}
Summarizing, we get, 
$|S^*|\geq m + \quad  \max_{{1\leq j\leq n-m}}\quad  \max\left(t(u_j) -m -j +1,0\right)$.
\\
We  show  that the algorithm TSS outputs a target set $S$ whose size 
is upper bounded by the value in  (\ref{eq-SK}). In general  $S$ does not consist of the 
nodes having the highest thresholds.
\\
Consider the residual graph $K(i)=(V_i,E_i)$, for some $1\leq i\leq n$. 
\remove{It is easy to see that   for any $u_j, u_s\in V_i$ it holds
  \begin{itemize}
	\item[] 1) $\delta_i(u_j)=i$;
	\\
	2) if $j<s$ then  $k_i(u_j)\leq k_i(u_s)$;
	\\
	3) if $t(u_j) \geq n$ then  $k_i(u_j) \geq i$,
\\
4) if $t(u_j) < n$ then  $k_i(u_j) \leq i$.
\end{itemize}}
W.l.o.g. we assume that  if the node to be selected is not unique then  the tie is broken as follows (notice that  if $j<s$ then  $k_i(u_j)\leq k_i(u_s)$):  

 \ i) If Case 1 holds then the selected node is the one  with the lowest index,

 ii) 
     otherwise the selected node is the one  with the largest index.

\noindent
This implies that $K(i)$ contains $i$ nodes with consecutive indices, 
that is, 

\centerline{$
V_i=\left\{u_{\ell_i},u_{\ell_i+1},\ldots,u_{r_i}\right\}
$
}

\noindent
for some $\ell_i\geq 1$ and $r_i=\ell_i+i-1$.

Let $h=n-m$. 
We   prove by induction on $i=n,\ldots, 1$, that
at the beginning of the $n-i+1$-th iteration of the while loop in TSS($K$), it holds 
{\small{
\begin{equation} \label{UBK}
|S\cap V_i| \leq 
\begin{cases}
{(r_i-h)+\max_{\substack{\ell_i \leq j \leq h}} \max( k_i(u_j)-(r_i-h)-j+\ell_i, \, 0)} & {\mbox{ if $r_i > h$,}}\\
{\max_{\substack{\ell_i \leq j \leq r_i} } \max(k_i(u_j)-j+\ell_i ,\, 0)} &{\mbox{ if $r_i \leq h$.}}
\end{cases} 
\end{equation}
}}
The  upper bound (\ref{eq-SK}) follows when  $i=n$; indeed $K(n)=K$ and $|S|=|S\cap V(n)|$.
\\
For $i=1$, the graph $K(1)$ consists of one node,  say $u$, and
$|S\cap\{u\}| $ takes value $1$ if if $k_1(u)\geq 1$, and value 0 if $k_1(u)=0$; hence
the bound holds in this case.\\
Suppose now (\ref{UBK}) true for some $i-1\geq 1$ and consider the $n-i+1$-th iteration of the algorithm TSS. 
Let $v$ be the node selected by algorithm TSS at the $n-i+1$-th iteration.
We distinguish three cases according to the algorithm cases and for each one, we can show that 
(\ref{UBK}) holds for the graph $K(i)$ (which is obtained by $K(i-1)$ by removing $v$).\hfill \qed
}

\section{Computational experiments.}

 We have extensively tested our algorithm TSS$(G)$ both on random graphs and on
real-world data sets, and we found that  our algorithm
performs surprisingly well in practice. This seems to suggest that the otherwise 
important inapproximability
result of Chen \cite{Chen-09} refers to  rare or artificial  cases.

\vspace*{-0.3truecm}
\subsection{ Random Graphs}

The first set of tests was done in order to compare the results of our algorithm to 
the exact solutions,  found
by formulating  the  problem as an 0-1 Integer Linear
Programming (ILP) problem. 
Although the ILP approach  provides  the optimal solution, 
it fails to return the solution in a reasonable time (i.e., days) 
already for moderate size networks.
We applied both our algorithm 
and  the  ILP algorithm to random graphs
with up to 50 nodes. 
{Figures \ref{fig4} depicts the results on Random Graphs  $G(n, p)$ on $n$ nodes (any possible edge occurs independently with probability $0 < p < 1$). The two plots report the results obtained for $n=30$ and $n=50$. For each plot the value of the $p$ parameter appears along the X-axis, while the size of the solution appears
along the Y-axis. Results on intermediates sizes exhibit similar behaviors.}
Our algorithm   produced  target sets of size close to the optimal (see Figure \ref{fig4});
for  several instances it  found an optimal solution.

\begin{figure}[ht!]
\begin{center}
\includegraphics[height=3.3truecm,width=5truecm]{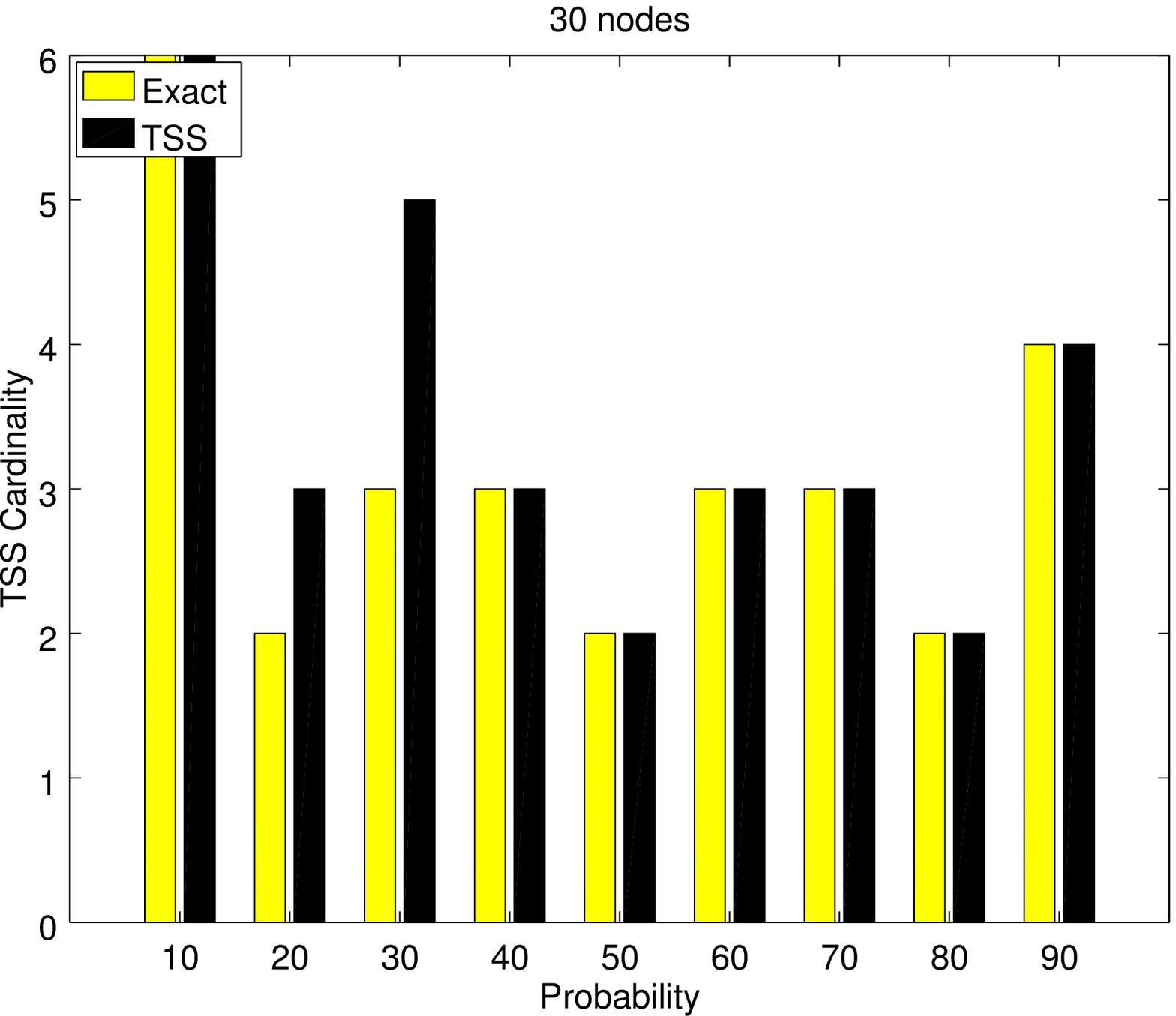}
$\qquad\qquad$
\includegraphics[height=3.3truecm,width=5truecm]{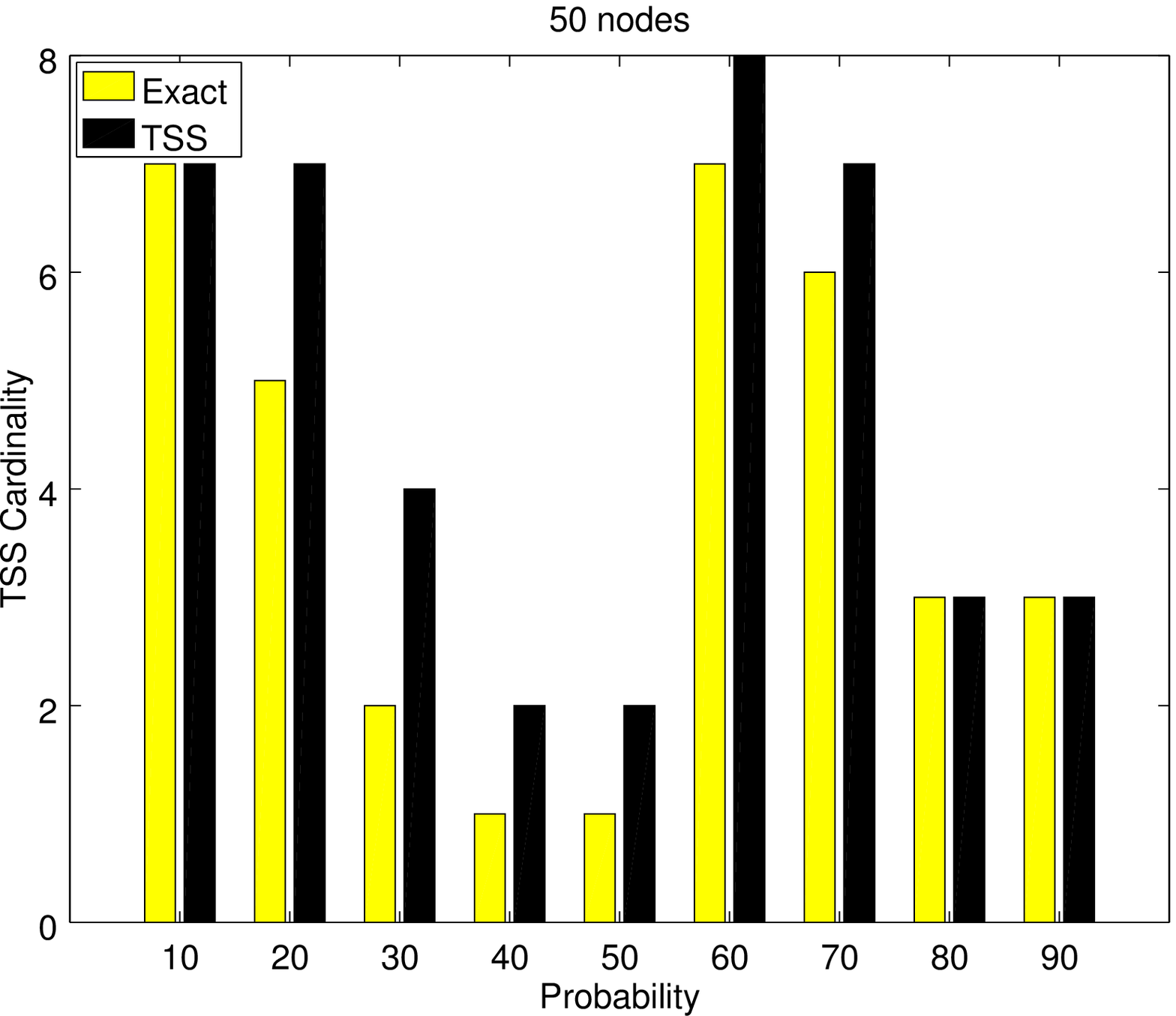}
\centerline{\hphantom{aaaaaaaaaaaaaaaaaaaaa} (a) \hphantom{aaaaaaaaaaaaaaaaaaaaaaaaaaaaaaaaaa} (b) \hphantom{aaaaaaaaaaaaaaaaaaaa}\hss}
\caption{Experiments  for random graphs  $G(n, p)$ on $n$ nodes (any possible edge occurs independently with probability $0 < p < 1$). 
 (a) $n=30$,   (b) $n=50$ with $p\in\{10/100, 20/100,\ldots, 90/100\}$.
 For each node the threshold was fixed to a random value between 1 and the node degree. \label{fig4} }
\end{center}
\end{figure}
\vspace*{-0.6truecm}
\subsection{Large Real-Life Networks }
We performed  experiments on several real social 
networks of various sizes from the Stanford Large Network Data set Collection (SNAP) \cite{snap}
and the Social Computing Data Repository at Arizona State University \cite{ZL09}.
The  data sets  we considered include both networks for which  small target sets exist
and networks needing larger target  sets (due to the existence of  communities, i.e., tightly connected disjoint groups of nodes
that appear to delay the diffusion process).

\paragraph{Test Network} Experiments have been conducted on the following networks: 
\begin{itemize}
	\item BlogCatalog \cite{ZL09}: a friendship network crawled from BlogCatalog, 
 a social blog directory website which manages the bloggers and their blogs. 
  It has  88,784 nodes and 4,186,390 edges. Each node represents a blogger and the network contains an edge $(u, v)$  if  blogger $u$ is friend of blogger $v$.
	\item BlogCatalog2 \cite{ZL09}: a friendship network crawled from BlogCatalog. 
  It has  97,884 nodes and 2,043,701 edges.
	\item BlogCatalog3 \cite{ZL09}: a friendship network crawled from BlogCatalog. 
  It has 10,312  nodes and 333,983 edges.
	\item BuzzNet \cite{ZL09}: BuzzNet is a photo, journal, and video-sharing social media network.	
  It has 101,168  nodes and 4,284,534 edges.
	\item CA-AstroPh\cite{snap}: A collaboration network of Arxiv ASTRO-PH (Astro Physics).    
It has  18,772 nodes and 198,110 edges. Each node represents an author  and the network contains an edge $(u, v)$  if an author $u$ co-authored a paper with author $v$.
		\item ca-CondMath \cite{snap}  A collaboration network of Arxiv COND-MAT (Condense Matter Physics).   
It has  23,133 nodes and 93,497 edges.
\item ca-GrQc \cite{snap}: A collaboration network of Arxiv GR-QC (General Relativity and Quantum Cosmology), 
It has  5,242 nodes and 14,496 edges.
\item ca-HepPh \cite{snap}: A collaboration network of Arxiv HEP-PH (High Energy Physics - Phenomenology),  it covers  papers from January 1993 to April 2003.   
It has  10,008 nodes and 118,521 edges. 
\item ca-HepTh \cite{snap}: A collaboration network of HEP-TH (High Energy Physics - Theory) 
It has  9,877 nodes and 25,998 edges. 
\item Delicious \cite{ZL09}: A friendship network  crawled on Delicious,  a social bookmarking web service for storing, sharing, and discovering web bookmarks.	It has	103,144	nodes and 1,419,519 edges.
\item Douban \cite{ZL09}: A friendship network crawled on  Douban.com, a Chinese  website providing user review and recommendations for movies, books, and music. It has  154,907 nodes  and  654,188 edges.
\item Lastfm \cite{ZL09}: Last.fm is a music website, founded in UK in 2002. It has claimed over 40 million active users based in more than 190 countries.	It has 	108,493	nodes  and   5,115,300 edges.
\item Livemocha \cite{ZL09}:	Livemocha is the world's largest online language learning community, offering free and paid online language courses in 35 languages to more than 6 million members from over 200 countries around the world.	It has 	104,438	nodes and 2,196,188 edges.
\item YouTube2 \cite{snap}:  is a data set crawled from YouTube,
the  video-sharing web site that includes a social network. In the Youtube social network, users form friendship each other and users can create groups which other users can join.
It contains  1,138,499 users and 2,990,443  edges.
\end{itemize}

The main characteristics of the studied networks are shown in Table 1. In particular, for each network we report  the maximum degree, the diameter, the size of the largest connected component (LCC), the number of triangles, the clustering coefficient and the network modularity \cite{N06}.
 
\begin{table}[ht]
\begin{center}
\resizebox{1.05\linewidth}{!} {
\begin{tabular}{|l|r|r|r|r|r|r|}
\hline
Name 											 & Max deg & Diam  &  LCC size & Triangles & Clust Coeff & Modul. \\ \hline
BlogCatalog \cite{ZL09}        &  9444     &     --     &     88784       &  51193389 &  0.4578          &      0.3182	\\ \hline
BlogCatalog2 \cite{ZL09}   		  &    27849  &   5       &      97884      &   40662527&  0.6857          &      0.3282       \\ \hline
BlogCatalog3 \cite{ZL09}   				&    3992   &   5       &      10312      &   5608664 &  0.4756          &      0.2374      \\ \hline
BuzzNet \cite{ZL09} 							&    64289  &   --       &      101163     &   30919848&  0.2508          &      0.3161     \\ \hline
ca-AstroPh \cite{snap} 							&    504   	&   14      &      17903      &   1351441 &  0.6768          &      0.3072      \\ \hline
ca-CondMath \cite{snap} 				&    279   	&   14      &      21363      &   173361  &  0.7058          &      0.5809      \\ \hline
ca-GrQc \cite{snap} 							&    81    	&   17      &      4158       &   48260   &  0.6865          &      0.7433     \\ \hline
ca-HepPh \cite{snap} 								&    491    &   13      &      11204      &   3358499 &  0.6115          &      0.5085     \\ \hline
ca-HepTh \cite{snap}						&    65    	&   17      &      8638       &   28399   &  0.5994          &      0.6128     \\ \hline
Delicious		\cite{snap}	       		 &   3216		  &   --     &     536108	     &   487972  &  0.0731          &      0.602     \\ \hline
Douban \cite{ZL09}						&    287    &   9       &     154908      &   40612   &  0.048           &      0.5773     \\ \hline
Last.fm \cite{ZL09}	 			 		&    5140		&   --       &      1191805    &   3946212 &    0.1378        &       0.1378     \\ \hline
Livemocha \cite{ZL09}						&    2980		&   6       &      104103     &   336651  &  0.0582          &      0.36     \\ \hline
Youtube2 \cite{ZL09}				   		&    28754  &   --       &      1134890    &   3056537 &  0.1723          &       0.6506     \\ \hline
\end{tabular}
}
\caption{Networks parameters.} 
\end{center}
\label{table}	
\end{table}

\paragraph{The competing algorithms.} We compare the performance of our algorithm TSS  toward  that of the best, to our knowledge, computationally feasible algorithms
in the literature.
Namely, we compare to 
 Algorithm {\em TIP\_DECOMP} recently presented in
\cite{SEP}, in which nodes minimizing the difference between  degree and threshold  are pruned from the graph 
until a ``core''  set is produced.
 We also compare our algorithm to the {\em VirAds} algorithm presented in \cite{DZNT}.
Finally, we compare to  an (enhanced) 
{\em Greedy} strategy  (given in Figure \ref{GreedyAlg}), in which nodes of maximum degree are iteratively inserted in the
target set  and pruned from the graph. 
Nodes that remains with zero threshold are simply eliminated from the graph, until no node remains.

\begin{figure}[ht!]
\begin{center}
\begin{tabular}{l}
\hline\\
{\bf  Algorithm} GREEDY-TSS($G$)\\
{\bf Input:} A graph $G=(V,E)$ with thresholds $t(v)$ for $v\in V$.\\
 $S=\emptyset$ \\
   $U=V$ \\  
   {\bf for } each $v\in V $ {\bf do} \{ \\
   \quad $\d( v)=d( v)$\\
  \quad$k( v)=t( v)$ \\
  \quad $N( v)=\Gamma( v)$ \\
	\} \\
 {\bf while} $U\neq \emptyset$  {\bf do }\{ \\
 
\quad $v={\tt argmin_{u\in U}\left\{k(u)\right\}}$ \\
\quad {\bf if } $k(v)>0$ {\bf then} \{\\
       $\quad$$\quad$       
				  	$v={\tt argmax}_{u\in U}\left\{\delta(u)\right\}$\\
				   $\quad$$\quad$     $S=S\cup\{v\}$   \\ \quad  \} \\
 $\ $ {\bf for } each $u\in N(v)$ {\bf do}  $\quad$ \{  \\
    
		 $\quad$$\quad$  $k(u)=\max\{0,k(u)-1\}$\\
		$\quad$$\quad$    $\d(u)=\d(u)-1$\\
   $\quad$$\quad$    $N(u)=N(u)-\{v\}$\\
$\quad$$\quad$     $U=U-\{v\}$\\
$\quad$\} \\
\} \\
\hline
\end{tabular}
\caption{GREEDY-TSS($G$)\label{GreedyAlg}}
\end{center}
\end{figure}

\paragraph{Thresholds values.} According to the scenario considered in  \cite{SEP}, in our experiments the thresholds are constant among all vertices (precisely the
constant value is an integer in the interval $[1, 10]$ and for each vertex $v$ the threshold
$t(v)$ is set as $min\{t, d(v)\}$ where $t = 1,2,\ldots,10$.

\paragraph{Results.} 
Figures \ref{fig9}--\ref{fig19} depict
the experimental results on  large real-life networks. For each network the results are reported
in a separated plot. For each plot the value of the threshold parameter
appears along the X-axis, while the size of the solution appears
along the Y-axis.
For each dataset, we compare the performance of our algorithm TSS   to the 
 algorithm {\em TIP\_DECOMP}
\cite{SEP},  to the  algorithm {\em VirAds} \cite{DZNT}, and to the 
{\em Greedy} strategy.

All test results consistently
show that the  TSS algorithm we introduce  in this paper    presents  the best performances  on all the considered networks, while none among  TIP\_DECOMP, VirAds, and  Greedy  is always better than the other two.


\begin{figure}[ht!]
\begin{center}
\includegraphics[width=10.5truecm]{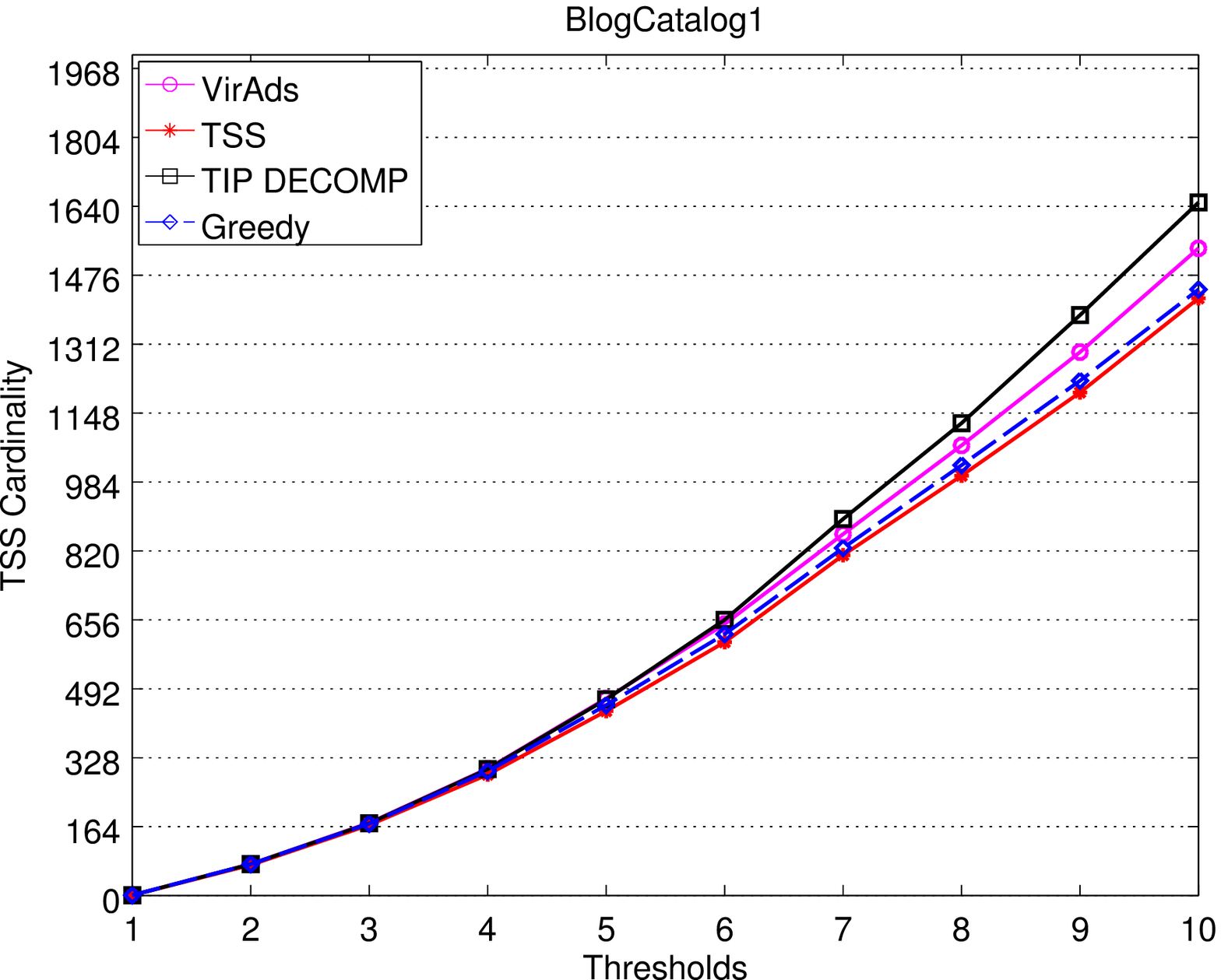}
\caption{ BlogCatalog \cite{ZL09}. \label{fig9} }
\end{center}
\end{figure}

\begin{figure}[ht!]
\begin{center}
\includegraphics[width=9.5truecm]{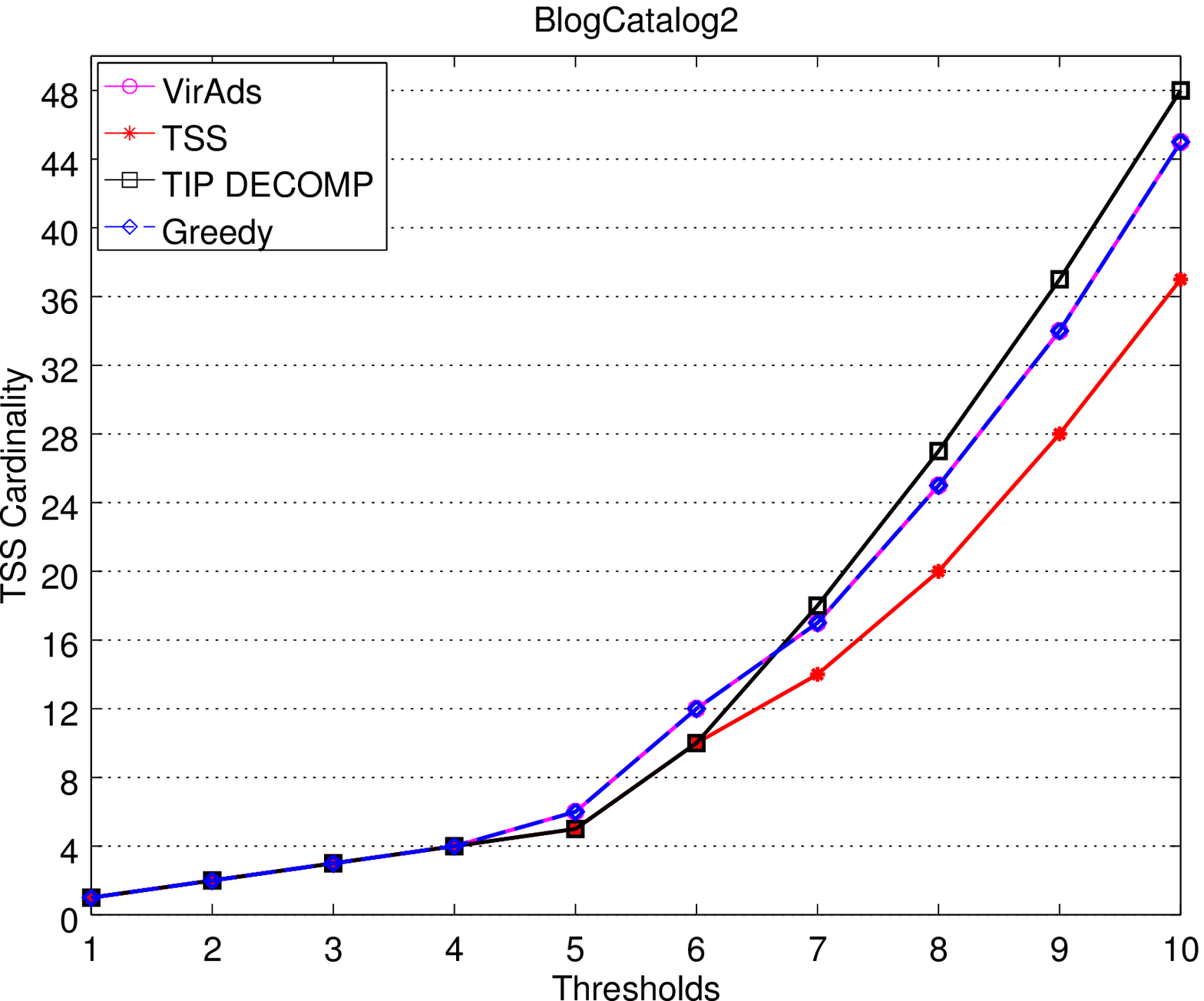}
\caption{BlogCatalog2 \cite{ZL09}.  \label{fig5} }
\end{center}
\end{figure}

\begin{figure}[ht!]
\begin{center}
\includegraphics[width=10truecm]{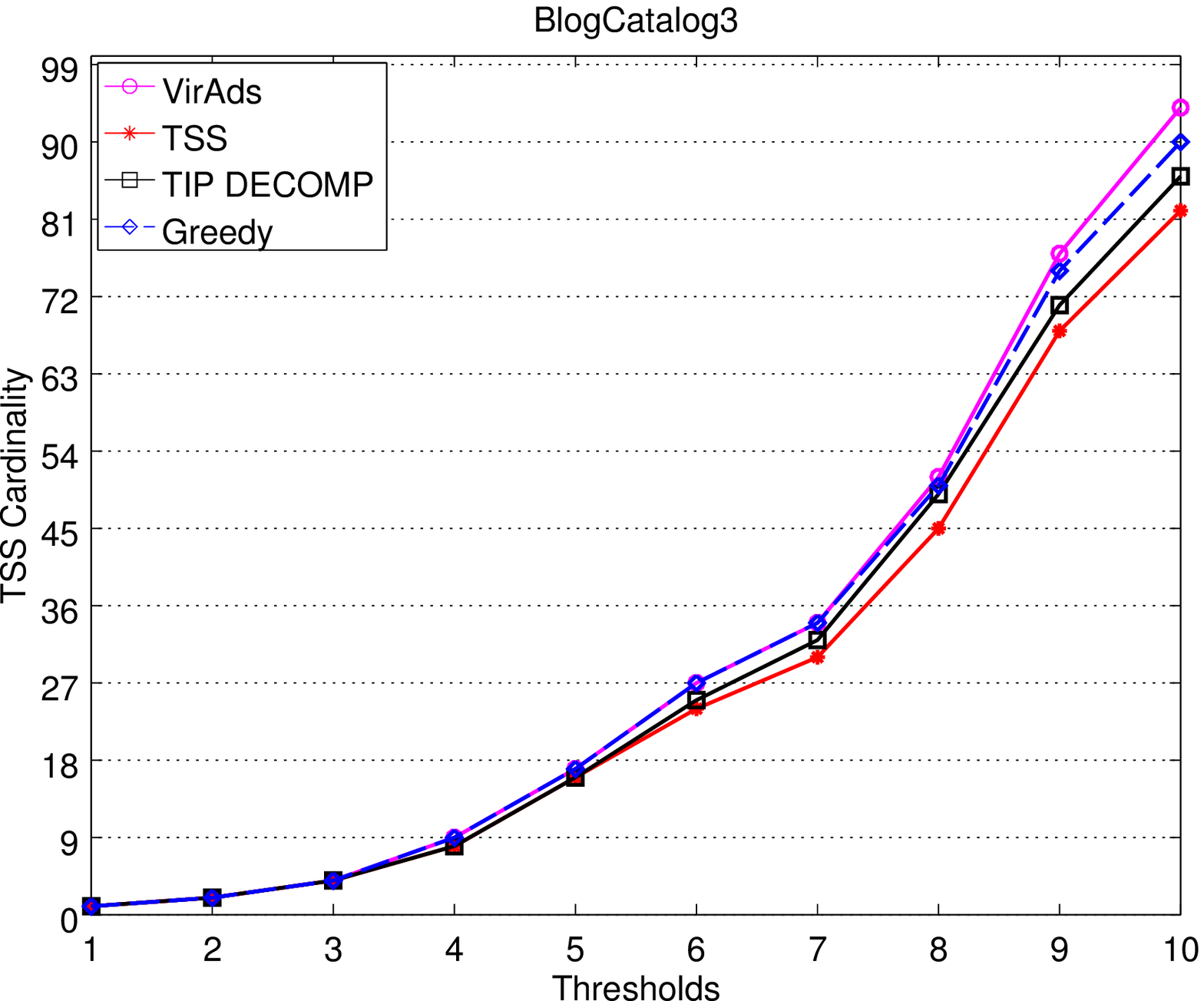}
\caption{ BlogCatalog3 \cite{ZL09}. \label{fig10} }
\end{center}
\end{figure}

\begin{figure}[ht!]
\begin{center}
\includegraphics[width=10truecm]{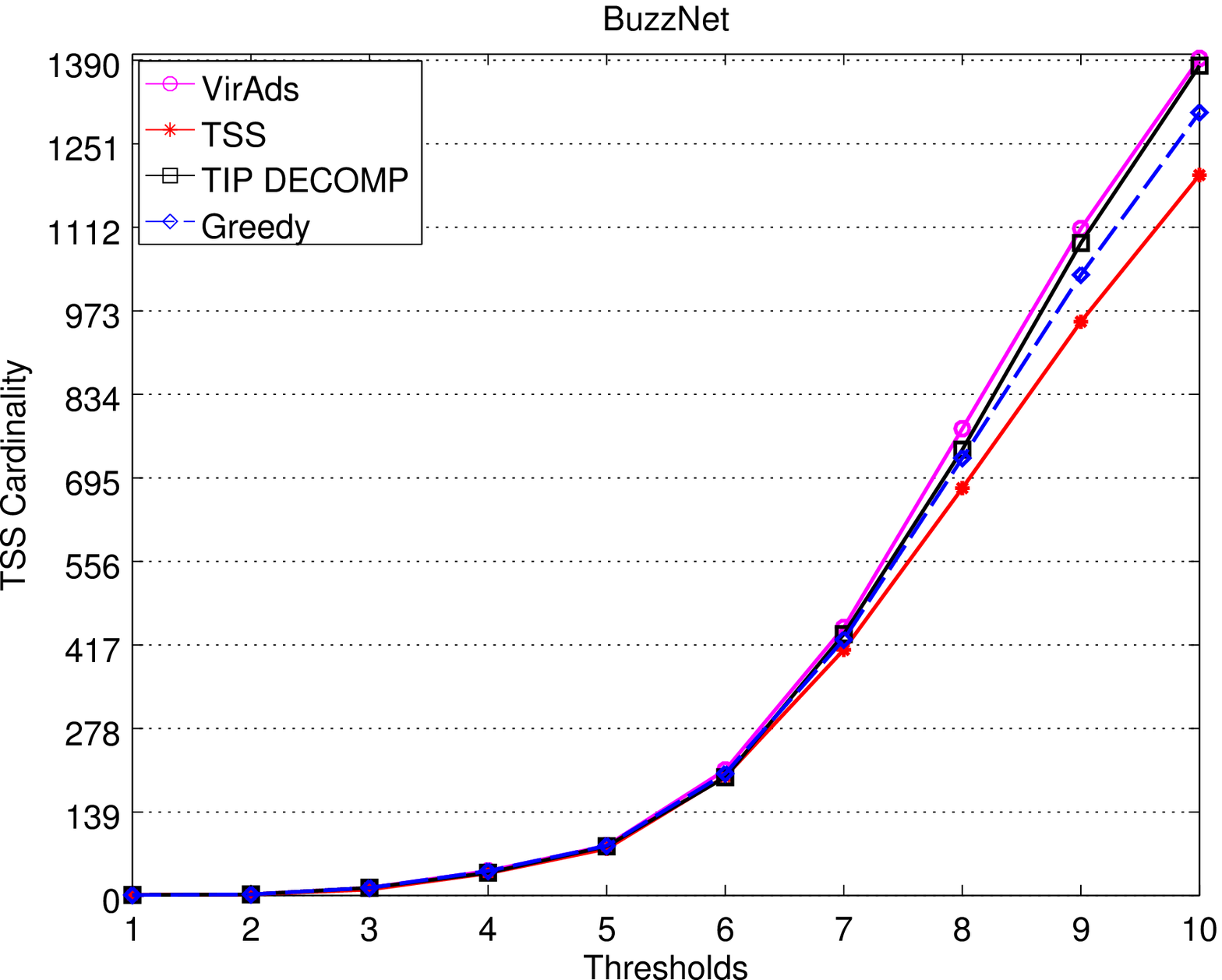}
\caption{ BuzzNet \cite{ZL09}. \label{fig11} }
\end{center}
\end{figure}	

\begin{figure}[ht!]
\begin{center}
\includegraphics[width=10truecm]{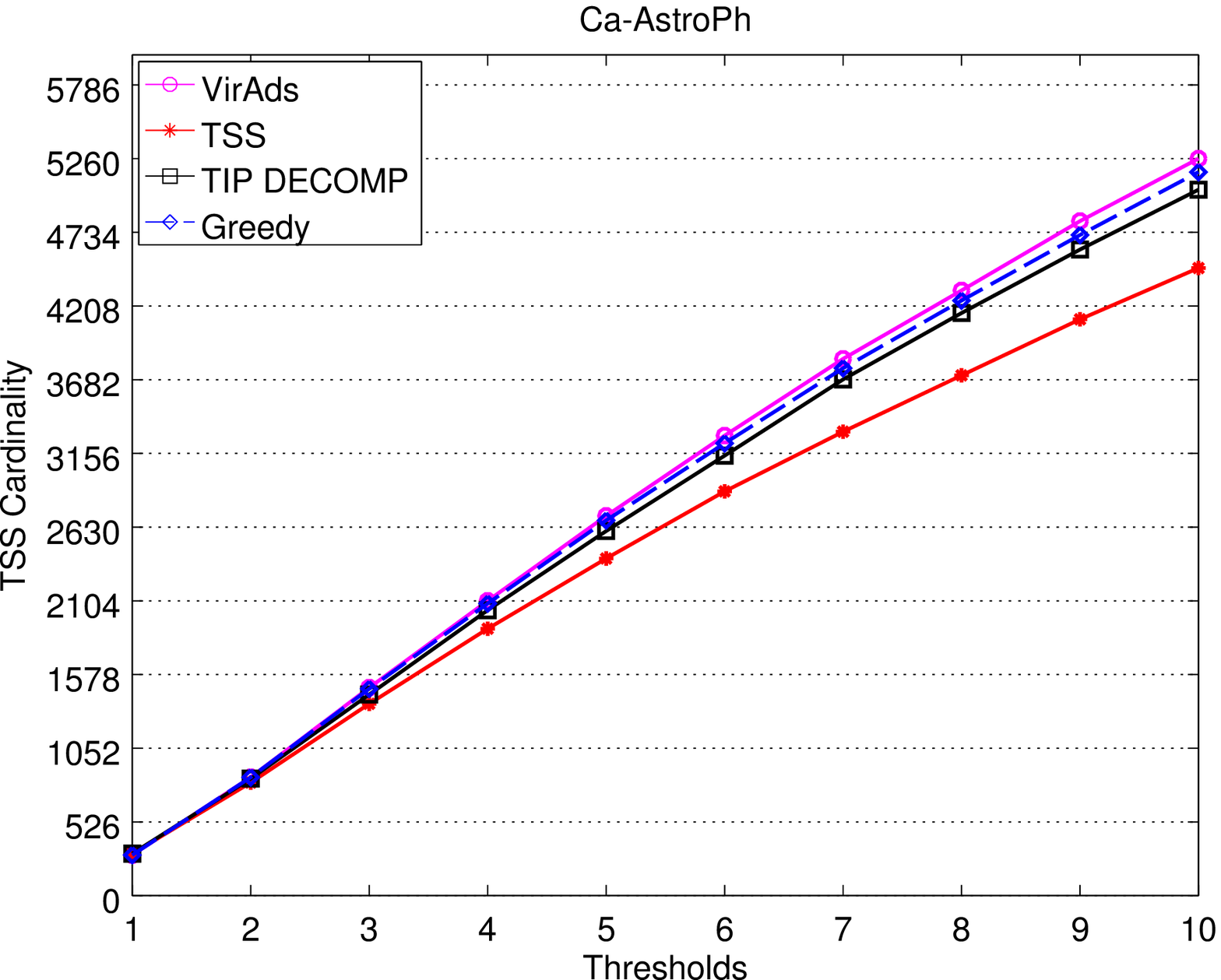}
\caption{CA-Astro-Ph\cite{snap}. \label{fig12} }
\end{center}
\end{figure}	

\begin{figure}[ht!]
\begin{center}
\includegraphics[width=10truecm]{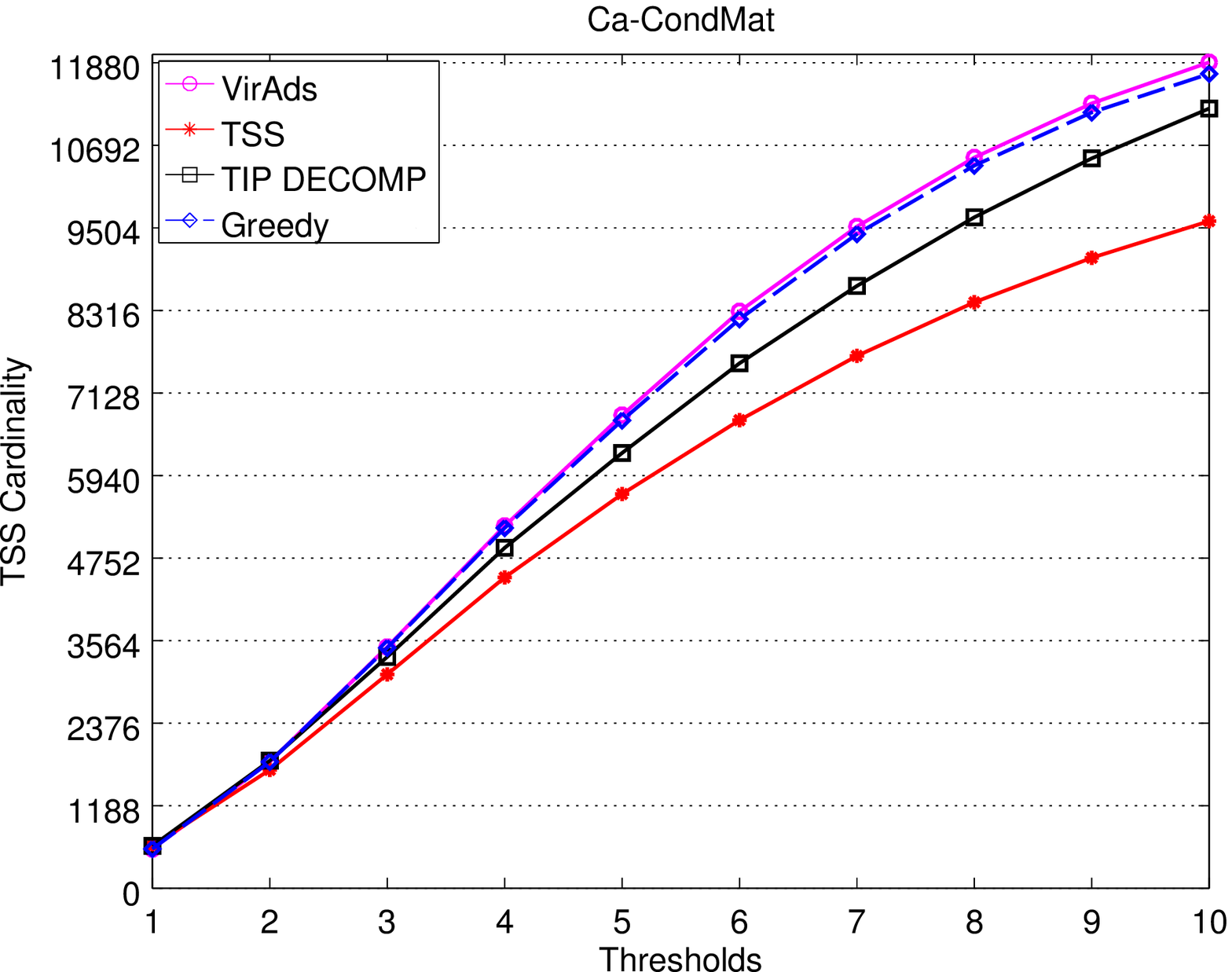}
\caption{ Ca-CondMat \cite{snap}. \label{fig8} }
\end{center}
\end{figure}

\begin{figure}[ht!]
\begin{center}
\includegraphics[width=10truecm]{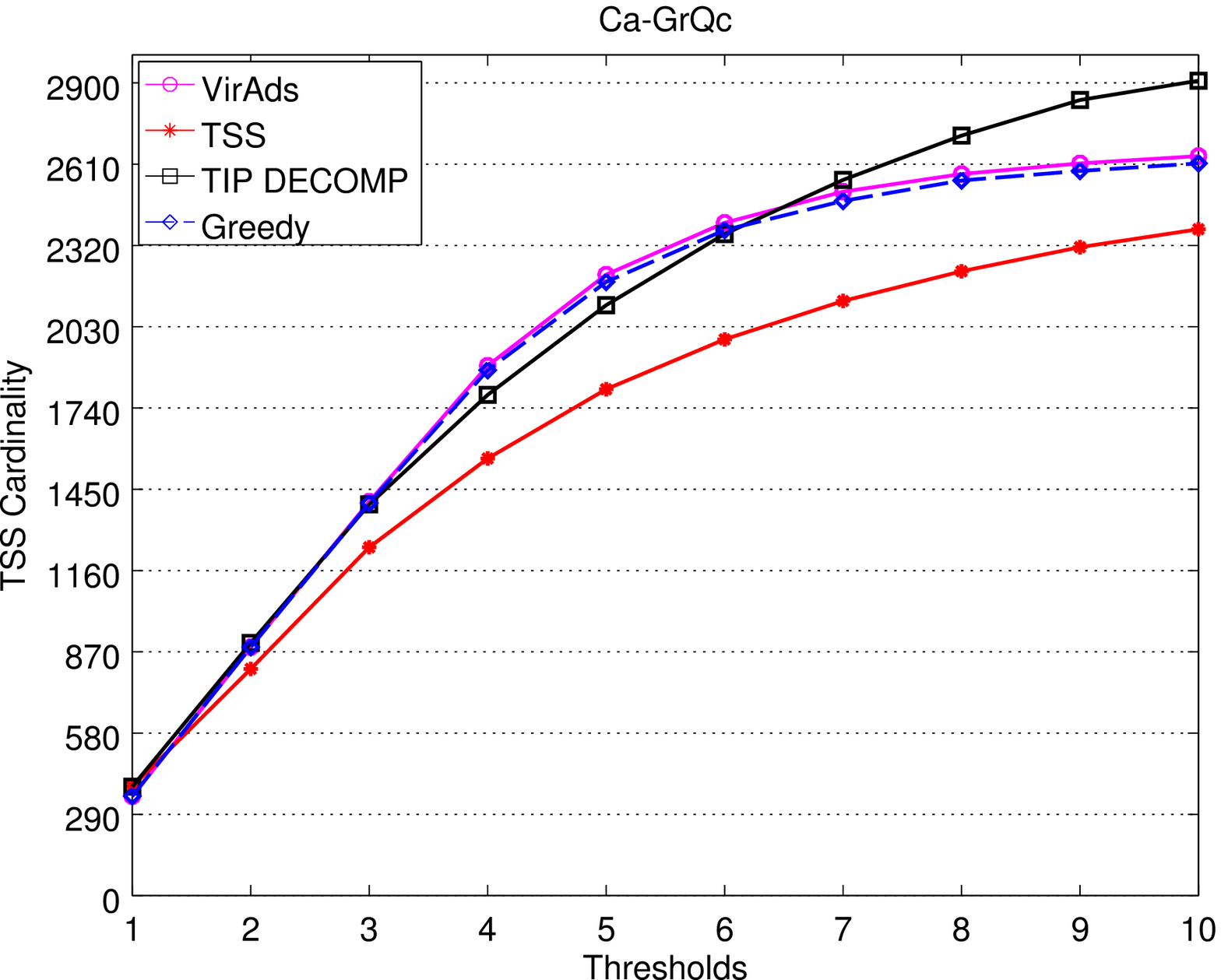}
\caption{ CA-GR-QC \cite{snap}. \label{fig13} }
\end{center}
\end{figure}	

\begin{figure}[ht!]
\begin{center}
\includegraphics[width=10truecm]{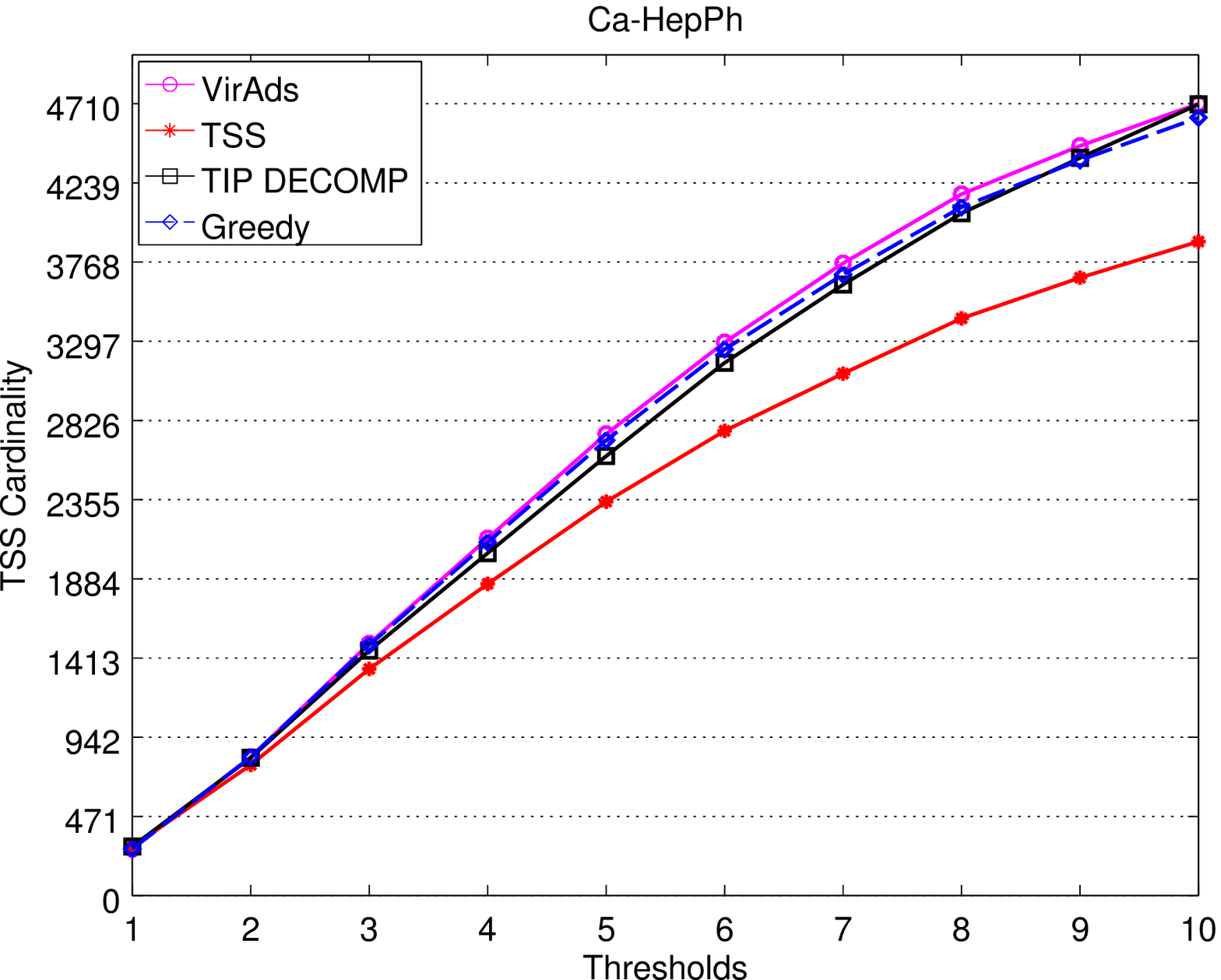}
\caption{ CA-HepPh \cite{snap}.  \label{fig6} }
\end{center}
\end{figure}

\begin{figure}[ht!]
\begin{center}
\includegraphics[width=10.2truecm]{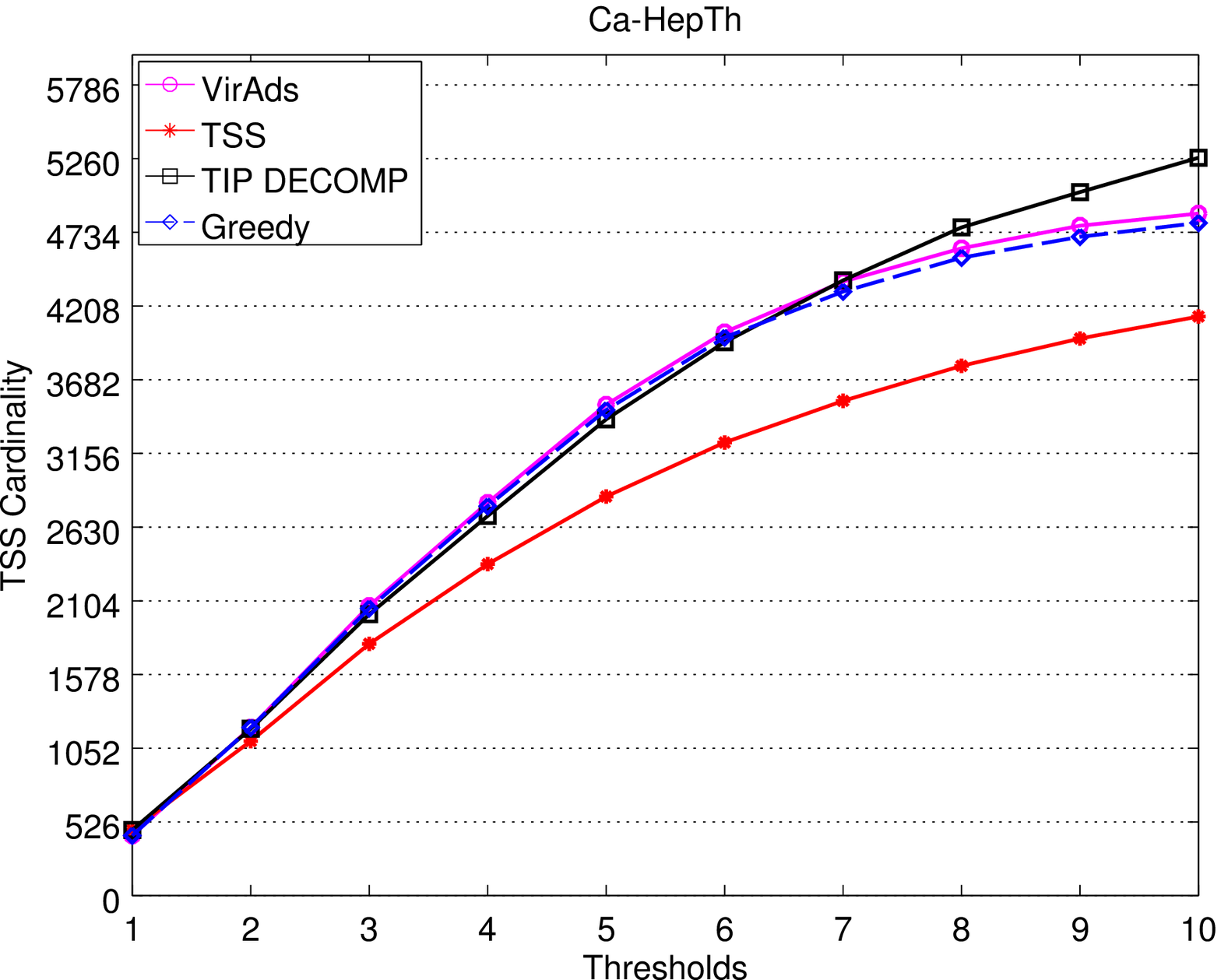}
\caption{Ca-HepTh \cite{snap}. \label{fig14} }
\end{center}
\end{figure}	

	\begin{figure}[ht!]
\begin{center}
\includegraphics[width=10.2truecm]{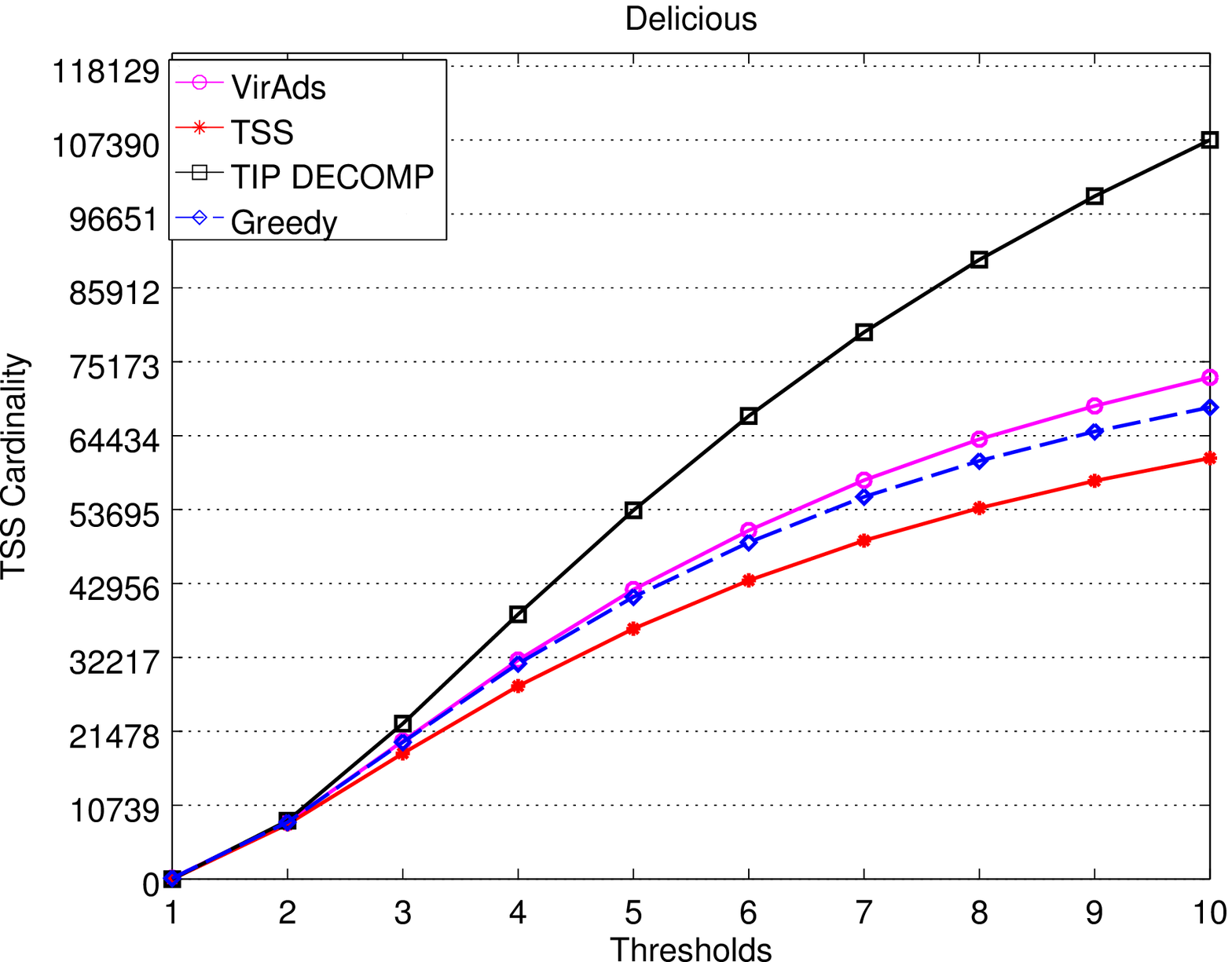}
\caption{Delicious \cite{ZL09}.  \label{fig7} }
\end{center}
\end{figure}

\begin{figure}[ht!]
\begin{center}
\includegraphics[width=10.2truecm]{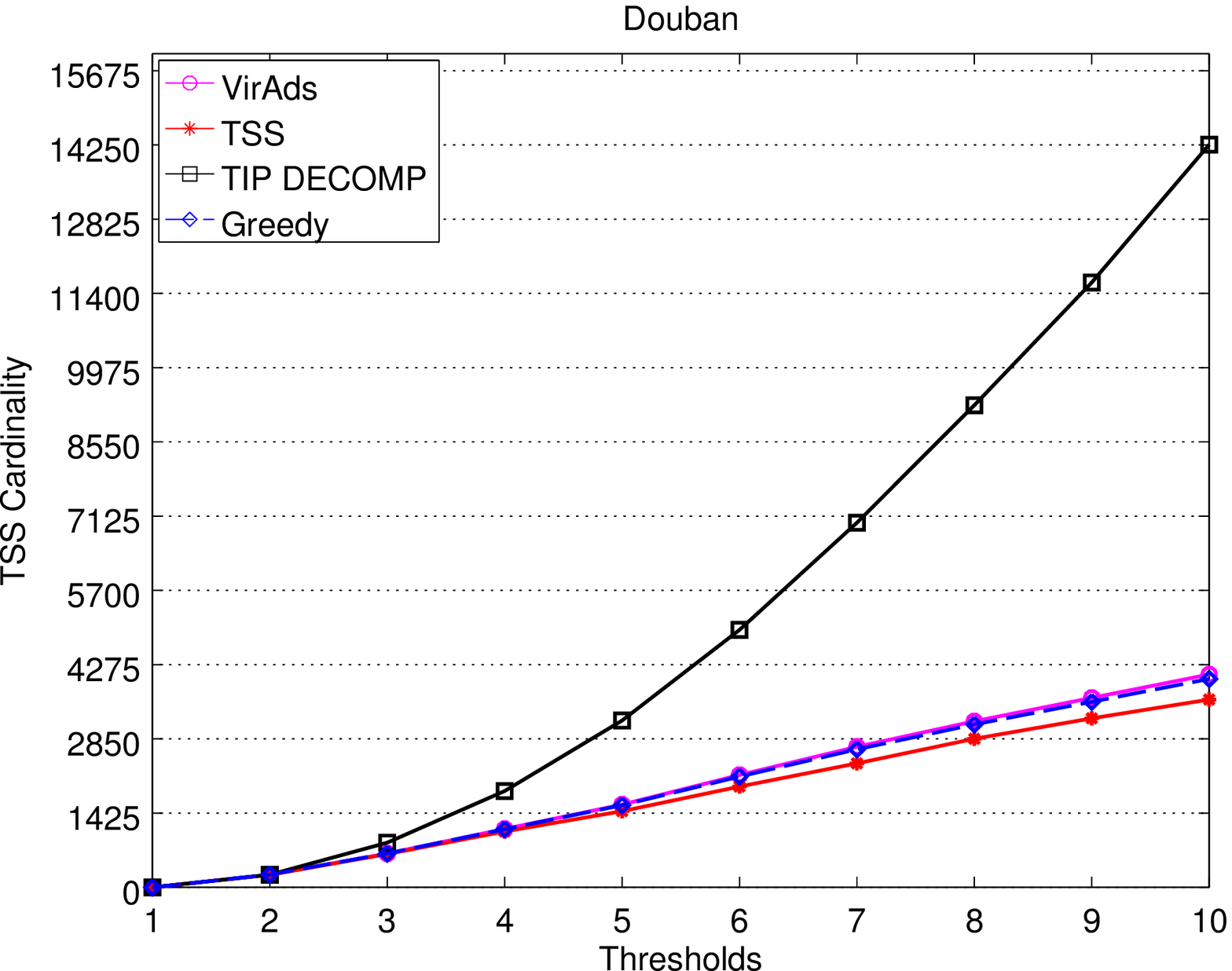}
\caption{Douban \cite{ZL09}.\label{fig15} }
\end{center}
\end{figure}	

\begin{figure}[ht!]
\begin{center}
\includegraphics[width=10truecm]{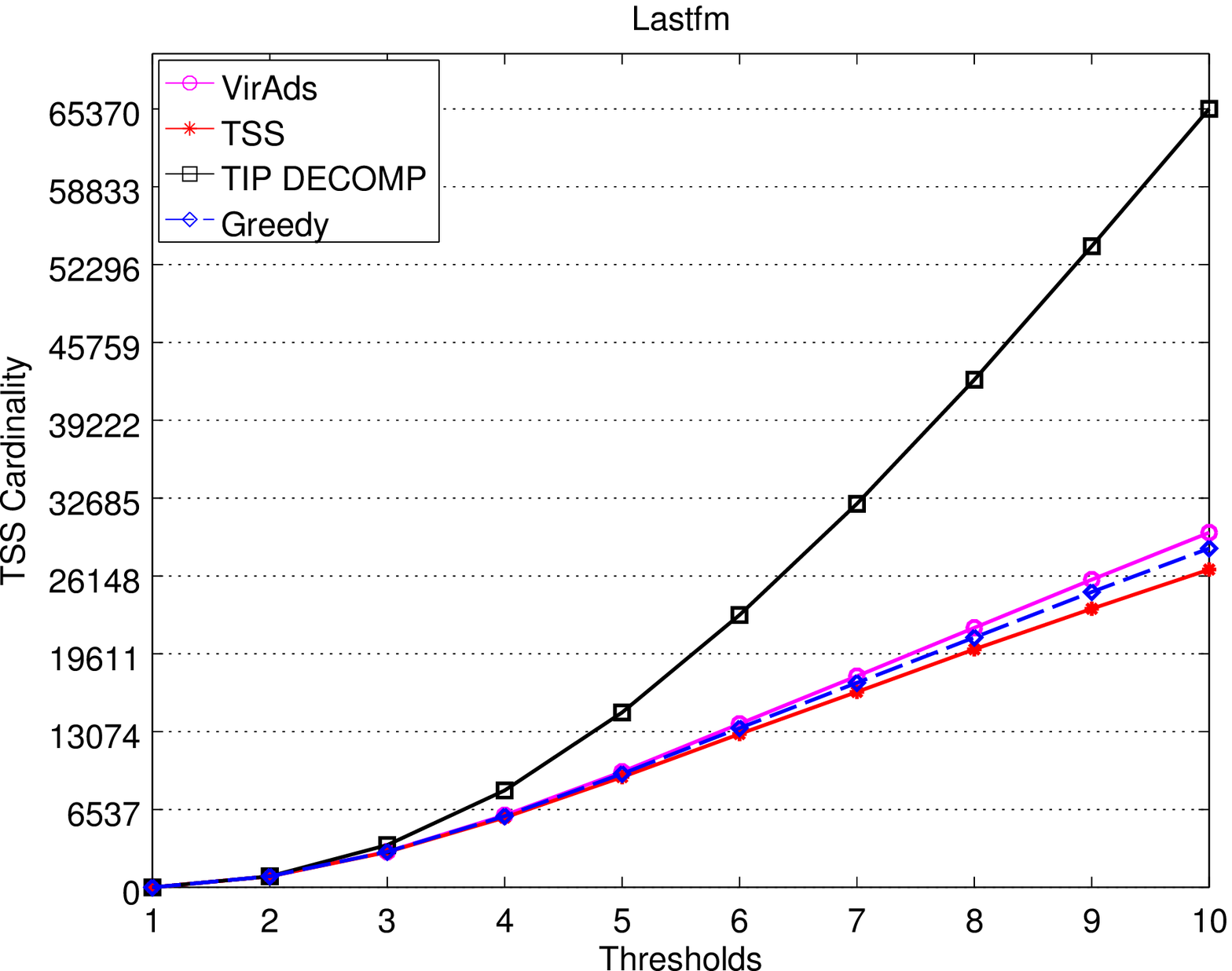}
\caption{Lastfm \cite{ZL09}.\label{fig17} }
\end{center}
\end{figure}

\begin{figure}[ht!]
\begin{center}
\includegraphics[width=10truecm]{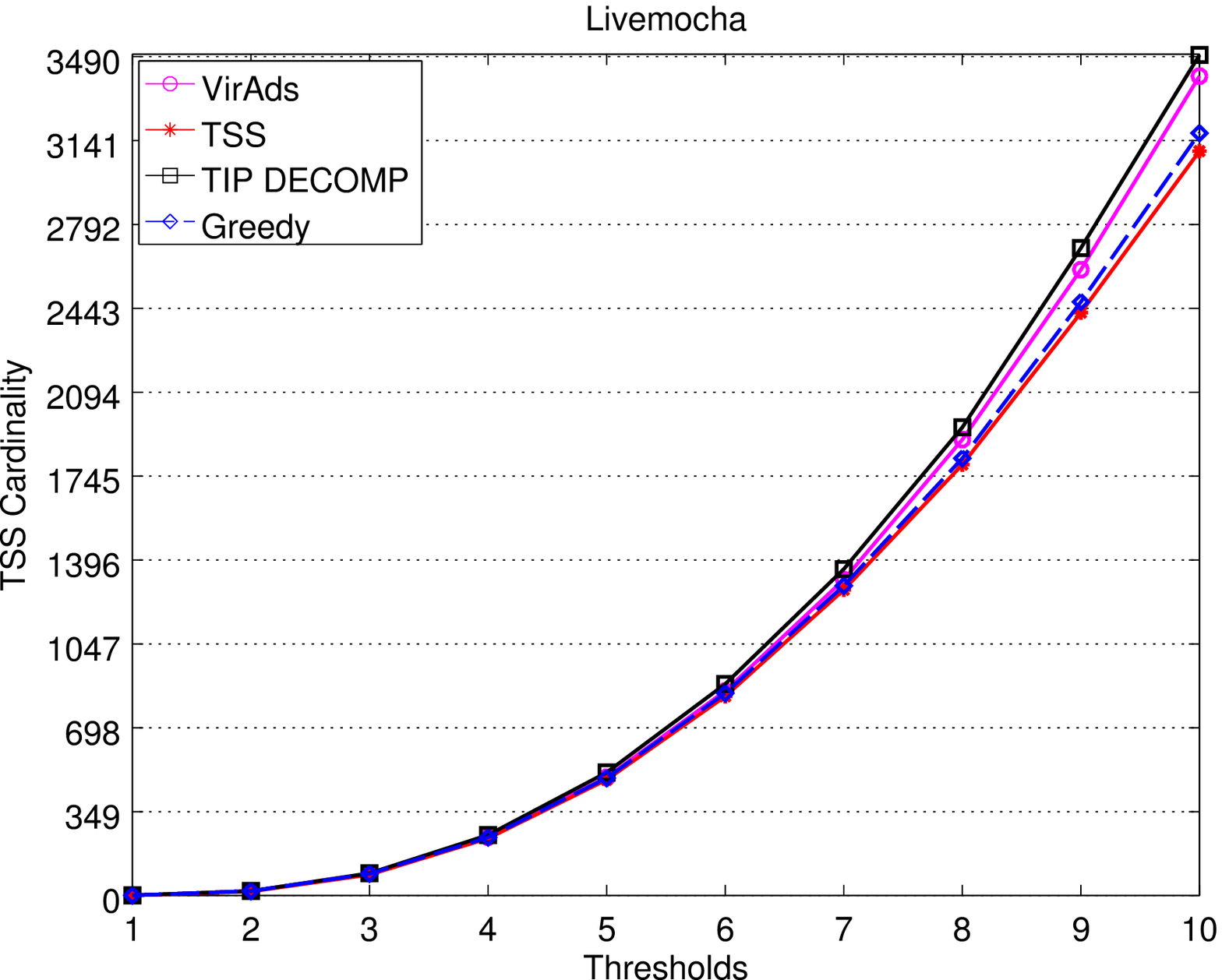}
\caption{Livemocha \cite{ZL09}.\label{fig18} }
\end{center}
\end{figure}
		
\begin{figure}[ht!]
\begin{center}
\includegraphics[width=9.7truecm]{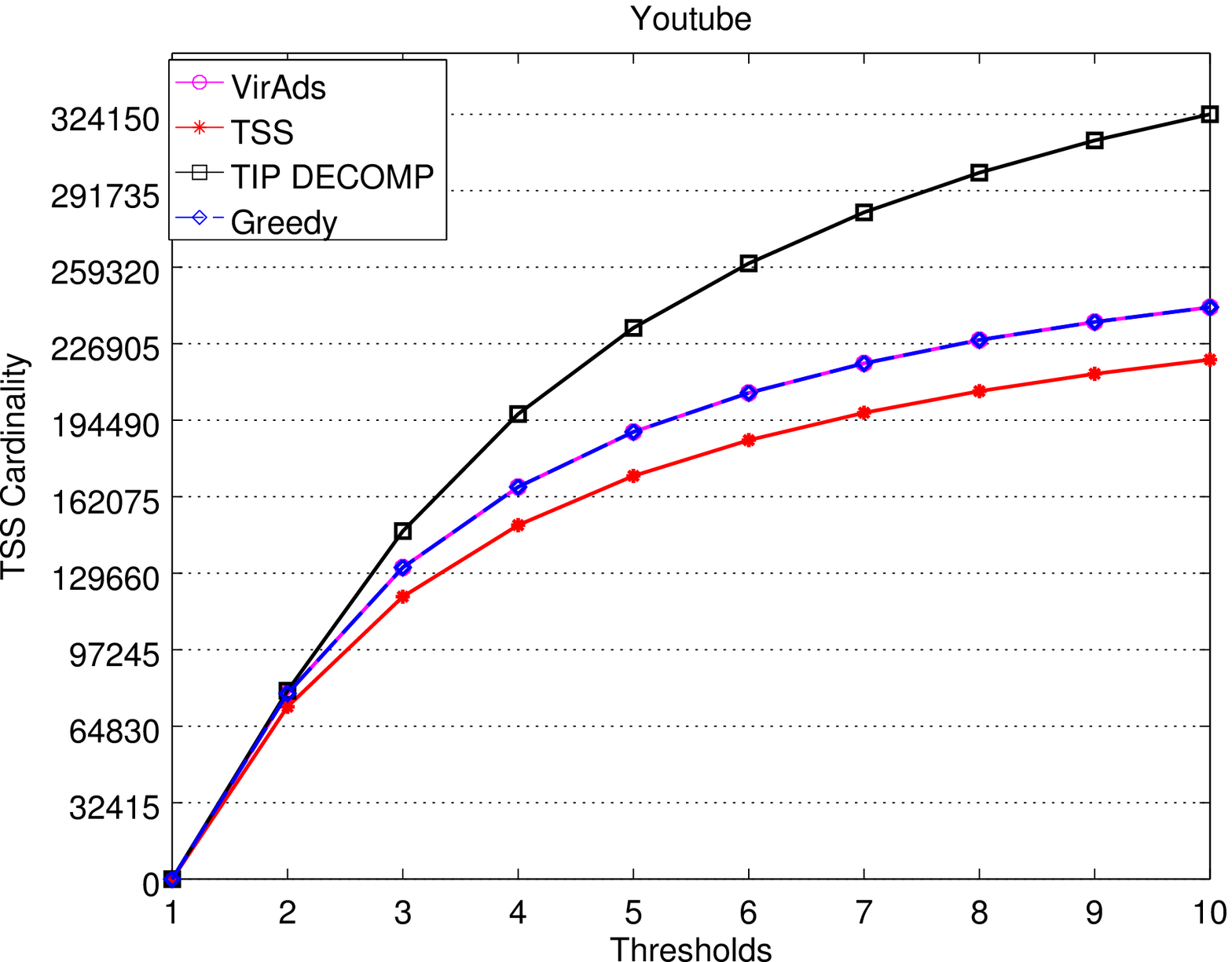}
\caption{YouTube2 \cite{snap}.\label{fig19} }
\end{center}
\end{figure}

\newpage

  \section{Concluding Remarks}
We  presented a simple algorithm to find small
sets of nodes that influence a whole network, 
where the dynamic that
governs the spread of influence in the network
is given in Definition 1.
In spite of its simplicity, 
our algorithm is optimal for several
classes of graphs, it {improves on } 
the general upper bound 
given in \cite{ABW-10} on the cardinality of a minimal influencing set, and outperforms, on real
life networks, the performances of known heuristics for
the same problem. There are many possible ways of extending our work.
We would be especially interested in discovering
additional interesting  classes of graphs for
which our algorithm is optimal (we conjecture that this is indeed the case).

\nocite{Cic+,Cic14,CGM+,CordascoGR15,CGRV16,CordascoGRV16,CGRV2015,Ga+,G,SNAM}
\bibliography{TSS-doppiasoglia}{}

\end{document}